\documentclass[11pt]{article}
\usepackage{jheppub}
\usepackage{tensor}
\usepackage{cleveref}
\usepackage{framed}
\usepackage{amsmath,amssymb,amsfonts,graphicx,tensor}
\usepackage{palatino}
\usepackage{newpxmath}
\usepackage{bbold}
\usepackage{subcaption}
\usepackage{graphicx}
\usepackage{enumitem}
\usepackage{braket}
\usepackage[dvipsnames]{xcolor}
\usepackage{comment}

\newcommand{\tr}{{\rm tr}}

\usepackage{amsthm,mathtools}

\newtheorem{theorem}{Theorem}

\newtheorem{proposition}[theorem]{Proposition}
\newtheorem{corollary}[theorem]{Corollary}
\theoremstyle{definition}
\newtheorem{definition}{Definition}

\newcommand{\id}{\mathrm{id}}

\usepackage{bbm}
\usepackage{float}
\usepackage{minted}
\usepackage{tcolorbox}

\def\Tr{{\rm Tr }}

\title{Factor code networks}
\author{Abhisek Sahu, Jeremy van der Heijden, Mark Van Raamsdonk, Rana Zibakhsh}
\emailAdd{abhi@phas.ubc.ca}
\emailAdd{jeremy.vanderheijden@ubc.ca}
\emailAdd{mav@phas.ubc.ca}
\emailAdd{rzibakhsh@phas.ubc.ca}
\affiliation{Department of Physics and Astronomy, University of British Columbia,\\
6224 Agricultural Road, Vancouver, B.C., V6T 1Z1, Canada}

\date{}

\abstract{For a quantum code defined by an isometry $T:\mathcal{H}_{\rm code} \to \mathcal{H}_{\rm phys}$, an operator map is a *-homomorphism $\varphi: B(\mathcal{H}_{\rm code}) \to B(\mathcal{H}_{\rm phys})$ from the space of operators on the logical Hilbert space to the space of operators on the physical Hilbert space satisfying $\varphi({\cal O}) T = T {\cal O}$ for every operator ${\cal O}$ on $\mathcal{H}_{\rm code}$. This can additionally be taken to be unital $\varphi(I_{code}) = I_{phys}$ if and only if $\dim \mathcal{H}_{\rm code}$ divides $\dim \mathcal{H}_{\rm phys}$. In this case, $\varphi(B(\mathcal{H}_{\rm code})$ is a von Neumann factor subalgebra of $B(\mathcal{H}_{\rm phys})$ and we call the pair $(T,\varphi)$ a {\it factor code}.

A code network is a family of Hilbert spaces $\{{\cal H}_i\}$ indexed by a partially ordered set with isometries $T_{ji}:{\cal H}_i \to {\cal H}_j$ for $i < j$ satisfying $T_{kj} T_{ji} = T_{ki}$ for $i < j < k$. In this paper, we ask when this can be promoted to a {\it factor code network} by defining unital operator maps $\varphi_{ji}$ so that $(T_{ji}, \varphi_{ji})$ is a factor code and  $\varphi_{kj} \varphi_{ji} = \varphi_{ki}$. We show that the divisibility conditions on Hilbert space dimensions are sufficient for any code network with a tree structure but not in general. For a diamond-shaped network, we show that promotion is possible if and only if the principle angles between $T_{42} \mathcal{H}_2$ and $T_{43} \mathcal{H}_3$ all have multiplicity divisible by $\mathcal{H}_1$. 

For incomparable $i_1,i_2 < j$, we characterize when $\varphi_{j i_1} (B(\mathcal{H}_{i_1}))$ and $\varphi_{j i_2} (B(\mathcal{H}_{i_2}))$ commute, when they generate all of $B(\mathcal{H}_j)$, and when they intersect only on multiples of the identity. For $i < j_1,j_2 < k$, we characterize when $\varphi_{k j_1} (B(\mathcal{H}_{j_1})) \cap \varphi_{k j_2} (B(\mathcal{H}_{j_2})) = \varphi_{k i} (B(\mathcal{H}_{i}))$ and when these algebras together with $B(\mathcal{H}_k)$ form a nondegenerate commuting square.
}

\begin{document}

\maketitle
\section{Introduction}

In a quantum code, a logical Hilbert space $\mathcal{H}_{\rm code}$ is isometrically mapped to a larger physical Hilbert space $\mathcal{H}_{\rm phys}$ so that each state in  $\mathcal{H}_{\rm code}$ is represented as a state in a ``code subspace'' of $\mathcal{H}_{\rm phys}$ isomorphic to $\mathcal{H}_{\rm code}$.\footnote{In standard theory of quantum error correction, an encoding is studied together with a noise model, a collection of correctable errors, and a recovery operation \cite{KnillLaflamme1997}. Throughout this paper, a quantum code means only an isometric encoding.}

Logical operators can also be represented as operators on the physical Hilbert space, but this representation is not unique. If 
\[
T:\mathcal{H}_{\rm code} \to \mathcal{H}_{\rm phys}
\]
is the isometric mapping to the code subspace, any mapping of operators\footnote{Here, $B(\mathcal{H})$ denotes the algebra of operators on $\mathcal{H}$.}
\[
\varphi: B(\mathcal{H}_{\rm code}) \to B(\mathcal{H}_{\rm phys})
\]
that respects the product and adjoint structure (i.e. is a $*$-homomorphism) and satisfies 
\[
\varphi({\cal O}) T = T {\cal O} \; 
\]
defines a valid representation of the logical operators on the physical Hilbert space. The last relation says that acting with the represented operator $\varphi({\cal O})$ after the mapping to the code subspace is exactly the same as acting with ${\cal O}$ before the mapping. 

For any such operator map $\varphi(B(\mathcal{H}_{\rm code}))$ defines a subalgebra of the full algebra $B(\mathcal{H}_{\rm phys})$ of physical operators. The use of observable algebras and their representations is central to operator-algebra quantum error correction \cite{BenyKempfKribs2007PRL,BenyKempfKribs2007PRA}. In the following, we will assume that all Hilbert spaces are finite-dimensional so that $B(\mathcal{H})$ is isomorphic to $M_n(\mathbb{C})$ with $n=\dim \mathcal{H}$.

\subsubsection*{Factor codes}

A natural additional property to demand for the operator map is that doing nothing to the logical state should correspond to doing nothing to the physical Hilbert space. This is the property that the operator map is {\it unital}, i.e. that $\varphi(I_{\rm code}) = I_{\rm phys}$.

As we show below, when the operator map is  unital, the algebra $\varphi(B(\mathcal{H}_{\rm code}))$ necessarily defines a (type-I) von Neumann factor subalgebra of $B(\mathcal{H}_{\rm phys})$; that is, we can represent $\mathcal{H}_{\rm phys}$ as $\mathcal{H}_{\rm code} \otimes \mathbb{C}^m$ such that $\varphi({\cal O}) = {\cal O} \otimes I_m$. In other words, the subalgebra $\varphi(B(\mathcal{H}_{\rm code}))$ is the algebra of all operators acting on a particular  tensor factor subsystem of the physical system.\footnote{Our construction is connected with subsystem and operator quantum error correction \cite{KribsEtAl2005,KribsEtAl2006} and with the idea that tensor-product structures can be specified by observable algebras \cite{ZanardiLidarLloyd2004}.} 
For this reason, we refer to a quantum code together with a unital operator map as a {\it factor code}.

A quantum code can be upgraded to a factor code if and only if $\dim \mathcal{H}_{\rm phys}$ is a multiple of $\dim \mathcal{H}_{\rm code}$. In this case, we will see that there is a family of possible unital operator maps. 

\subsubsection*{Code networks}

In certain contexts, it is interesting to consider a network of quantum codes where we have a set of Hilbert spaces $\{\mathcal{H}_i\}$ indexed by a partially ordered set and isometries $T_{ji}:\mathcal{H}_j \to \mathcal{H}_i$ for $i < j$ with $T_{kj} T_{ji} = T_{ki}$ for $i < j < k$. 
A particular motivating example is the situation where we have a tensor network such that the Hilbert spaces associated with certain regions of the the tensor network are mapped isometrically into the Hilbert spaces for larger regions that contain them. 

The central problem that we consider in this paper is when such a code network can be upgraded to a compatible network of factor codes by defining unital operator maps $\varphi_{ji}:B(\mathcal{H}_i) \to B(\mathcal{H}_j)$ satisfying
\[
\varphi_{ji} ({\cal O}) T_{ji} = T_{ji} {\cal O} \qquad \varphi_{kj} \circ \varphi_{ji} = \varphi_{ki} \; .
\]

When the network has a tree-like structure, we will show that the upgrade to a factor code network is always possible when all the isometries map a Hilbert space to another whose dimension is an integer multiple. However, when the network includes multiple distinct paths between the same two Hilbert spaces (i.e. elements $i < j_1, j_2 < k$ with $j_1,j_2$ incomparable), there can be an obstruction to choosing compatible unital operator maps even when the dimensional constraints are satisfied. Whether or not the construction is possible depends on the relative position of the two code subspaces $T_{k j_1} {\cal H}_{j_1}$ and $T_{k j_2} {\cal H}_{j_2}$ in ${\cal H}_k$.

\subsubsection*{Summary of results}

In section 2, we provide basic definitions and give several equivalent characterizations of a factor code. In section 3, we describe a general construction for factor codes based on promoting the isometry $T$ to a set of isometries $T^{(a)}$ to orthogonal subspaces of $\mathcal{H}_{\rm phys}$ that span $\mathcal{H}_{\rm phys}$. We characterize the space of possibilities for this construction, showing an equivalence with the quotient $U(nm - n)/U(m-1)$ where $n$ and $nm$ are the dimensions of the code and physical Hilbert space.

In section 4, we consider code-networks with a a tree-like structure, showing that these can always be upgraded to factor code networks provided the dimensional constraints are satisfied. For the special case of ``duplex codes" with two code subspaces $T_A \mathcal{H}_A$ and $T_B \mathcal{H}_B$ contained within a single physical Hilbert space ${\cal H}$, the operator maps in a factor code give two subalgebras $\mathcal{M}_A$ and $\mathcal{M}_B$ of $B({\cal H})$. We characterize when for either particular or generic choices of the operator maps, these algebras generate the full algebra $B({\cal H})$, intersect only on multiples of the identity (i.e. share no non-trivial observable), or commute.

In section 5, we consider diamond-shaped networks, giving necessary and sufficient conditions for when these can be promoted to factor code networks. In the case that they can, we have a nested set of algebras $\mathcal{M}_1 \subset \mathcal{M}_2, \mathcal{M}_3 \subset \mathcal{M}_4$ inside the algebra $\mathcal{M}_4$ associated to the top Hilbert space. We analyze when it is possible/generic to obtain the natural intersection property $\mathcal{M}_2 \cap \mathcal{M}_3 = \mathcal{M}_1$ and when the algebras can be chosen to form a nondegenerate commuting-square.

\subsubsection*{Related work}

Chains and trees of quantum codes arise in hierarchically composed encoders and concatenated-code constructions \cite{KnillLaflamme1996Concatenated}. However, our primary motivation is the relation between tensor networks and quantum codes \cite{FerrisPoulin2014}. Here, we can often associate Hilbert spaces to specific subsystems of a tensor network and in some cases, the tensor network provides a natural isometry between the Hilbert spaces associated with a region and some larger region that includes it. This feature is particularly relevant in tensor networks that serve as toy-models of holographic quantum gravity \cite{AlmheiriDongHarlow2015,PastawskiEtAl2015,HaydenEtAl2016,Harlow2017}. The present study arose specifically while asking whether boundary factor algebras could be assigned to specific regions of holographic tensor networks \cite{SahuEtAl2026}; these algebras would then provide the boundary representation of observables associated with some region of the tensor network, modeling the physics in a localized region of spacetime. The holographic tensor-network applications of our results are developed in the companion paper \cite{CompanionPaper}. The structures we study also resemble the structure of nets algebras associated with nested spatial regions in quantum field theories \cite{HaagKastler1964}. Our work in \cite{SahuEtAl2026} seeks to understand to what extent this algebraic structure in quantum field theory is preserved when gravity (quantum spacetime dynamics) effects are included.

\subsubsection*{Use of AI}

A number of the proof ideas and technical steps in this work (including some of the necessary and sufficient conditions appearing in the theorems) were suggested by Chat GPT and Google Gemini (Deep Mind). However, all of the detailed proofs were constructed and written by the authors. The full draft was also written by the authors, but Chat GPT and Google Gemini were used for final proofreading and checks of all proofs.

\section{Codes, networks, and factor codes}

We begin with some standard definitions.

\subsection{Quantum codes}

\begin{definition}
    A {\it quantum code} is an isometric embedding $T : \mathcal{H}_{\rm code} \to \mathcal{H}_{\rm phys}$ of a smaller Hilbert space into a larger Hilbert space. The isometry condition is $T^\dagger T = 1$. The image $T \mathcal{H}_{\rm code}$ is a subspace of $\mathcal{H}_{\rm phys}$ known as the {\it code subspace}. 
\end{definition}
For a quantum code, a state $|\psi\rangle$ in the logical Hilbert space $\mathcal{H}_{\rm code}$ is encoded as a state $T \ket{\psi}$ in the physical Hilbert space. For an operator ${\cal O}$ that acts on the logical Hilbert space, we can define a representative $\tilde{\cal O}$ acting on the physical Hilbert space. We demand that this satisfies
\[
\tilde{{\cal O}} T = T {\cal O} \;,
\]
so that for any state $|\psi \rangle \in {\cal H}_{\rm code}$, $\tilde{{\cal O}}$ acting on the encoded state $T|\psi \rangle$ gives the encoded version of ${\cal O} |\psi \rangle$. It is natural to choose operator representatives in a consistent way for the full set $B({\cal H}_{\rm code})$ of operators in the logical Hilbert space so that the algebra structure is preserved. Thus we define $\tilde{\cal O}$ via a homomorphism:
\begin{definition}
    An {\it operator map} for a quantum code $T : \mathcal{H}_{\rm code} \to \mathcal{H}_{\rm phys}$ is a $*$-homomorphism $\varphi:B({\cal H}_{\rm code}) \to B({\cal H}_{\rm phys})$ i.e. a linear map with 
    \[
    \varphi({\cal O}_1 {\cal O}_2) = \varphi({\cal O}_1) \varphi({\cal O}_2) \qquad \qquad \varphi({\cal O}^\dagger) = \varphi({\cal O})^\dagger 
    \] satisfying
\[
\varphi({\cal O}) T = T {\cal O} \; .
\]
\end{definition}
We do not require an operator map to be unital (i.e. obey $\varphi(I_{\rm code}) = I_{\rm phys}$). Later, we will be interested in understanding the special class of unital operator maps.

For a given quantum code, there are many choices for $\varphi$, but there is a canonical choice:
\[
\varphi(\mathcal{O}) = T {\cal O} T^\dagger \; .
\]
With this (generally non-unital) choice, all such operators annihilate any state orthogonal to the code subspace. 

Next, we consider a more general system with a consistent network of quantum codes. 
\begin{definition}
A {\it code network} is a family of Hilbert spaces
$\{\mathcal H_i\}_{i\in S}$ indexed by a partially ordered set $S$,
together with isometries
\[
T_{ji}:\mathcal H_i\to \mathcal H_j
\]
for every $i\le j$, satisfying $T_{ii}=I_{\mathcal H_i}$ and
\[
T_{ki}=T_{kj}T_{ji}
\]
whenever $i\le j\le k$.
\end{definition}
To visualize a code network, we can introduce a \emph{Hasse diagram} representing the partially ordered set $S$. This is a directed graph with a vertex for each element of ${\cal S}$ and a directed edge $i \to j$ between any two vertices $i$ and $j$ with $i < j$ and with no $k$ such that $i < k < j$. Figure \ref{fig:Hasse} shows examples of these diagrams. The Hasse diagram of Figure \ref{fig:Hasse} (left) has the special property of being a tree.

\begin{figure}
    \centering
    \includegraphics[width=0.5\linewidth]{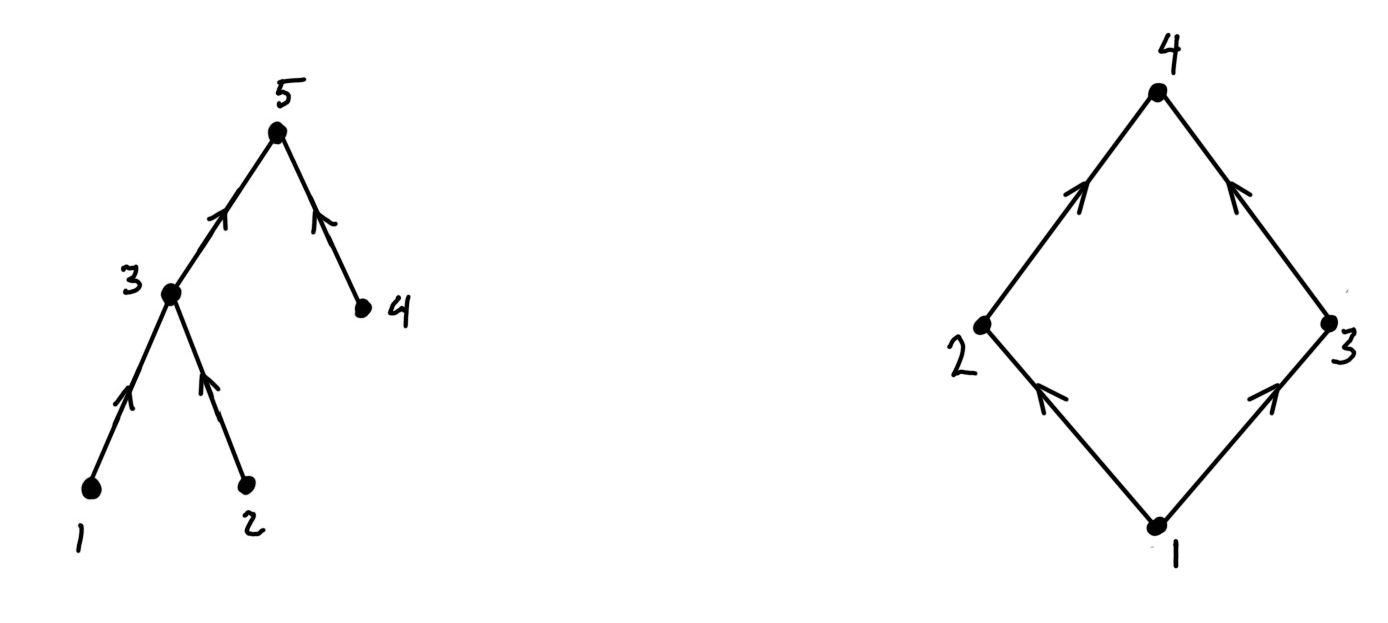}
    \caption{Hasse diagrams for code networks. Each vertex corresponds to a Hilbert space and each directed line corresponds to an isometry.}
    \label{fig:Hasse}
\end{figure}

For a code network, we can define a corresponding set of operator maps $\varphi_{ji}$ satisfying the properties in the definition above together with the consistency conditions
\[
\varphi_{kj} \circ \varphi_{ji} = \varphi_{ki}  \qquad i < j < k \; , \qquad \varphi_{ji}({\cal O}_i) T_{ji} = T_{ji} {\cal O}_i \; .
\]
The canonical choice $\varphi_{ji}({\cal O}) = T_{ji} {\cal O} T_{ji}^\dagger$ satisfies these properties.

Code networks naturally arise in certain tensor networks, where various bulk regions have associated Hilbert spaces and the tensor network with the region removed defines (for appropriate regions and tensors) an isometry to the boundary Hilbert space.

\subsection{Factor codes}

For the general quantum codes defined above, the image of ${\cal H}_{\rm code}$ under the isometric map $T$ is a subspace of ${\cal H}_{\rm phys}$ and the image of the algebra $B({\cal H}_{\rm code})$ under the canonical operator map $\varphi: {\cal O} \to T {\cal O} T^\dagger$ is a type of non-unital subalgebra of $B({\cal H}_{\rm phys})$ known as a {\it corner} since all its elements are of the form $P {\cal O} P$ for a projector $P = T T^\dagger$. This algebra does not include the identity operator of the full algebra unless $\mathcal{H}_{code}$ and $\mathcal{H}_{phys}$ are isomorphic; in particular, the identity operator on the logical Hilbert space maps to the projector $P$. 

In some cases, it is possible to choose the operator map such that the resulting subalgebra is unital, with $\varphi(I_{\rm code}) = I_{\rm phys}$. In these cases, as we show in the next Proposition, there is a tensor factorization $\mathcal{H}_1 \otimes \mathcal{H}_2$ of the physical Hilbert space such that the code subspace is identified with $\mathcal{H}_1 \otimes e_1$ for some vector $e_1 \in \mathcal{H}_2$ and the image under $\varphi$ of the algebra on $\mathcal{H}_{\rm code}$ is the subalgebra ${\cal O}_1 \otimes I_2$ acting on the first factor. Algebraically, this structure is known as a {\it von Neumann factor}.\footnote{A von Neumann algebra $\mathcal{M}\subset B(\mathcal{H}_{\rm phys})$ is called a factor if there are no non-trivial operators in $\mathcal{M}$ that commute with all operators in $\mathcal{M}$, i.e. $\mathcal{M}'\cap \mathcal{M} =\mathbb{C} I_{\rm phys}$ where $\mathcal{M}' = \{\mathcal{O} \in B(\mathcal{H}_{\rm phys})\,|\, [\mathcal{O}, \mathcal{A}]=0~ \, \forall \mathcal{A} \in \mathcal{M} \} $. Note that since a von Neumann algebra satisfies  $\mathcal{M}=\mathcal{M}''=(\mathcal{M}')'$ it automatically contains the identity $I_{\rm phys} \in \mathcal{M}$.} This motivates the following definition:
\begin{definition}
    A {\it factor code} is a quantum code $T:\mathcal{H}_{\rm code} \to \mathcal{H}_{\rm phys}$ together with a unital *-homomorphism $\varphi: B(\mathcal{H}_{\rm code}) \to B(\mathcal{H}_{\rm phys})$ satisfying
    \[
    \varphi({\cal O}) T = T {\cal O}  \; .
    \]
\end{definition}
The following proposition describes four equivalent characterizations of a factor code:
\begin{proposition}
\label{factor}
    For a quantum code $T : \mathcal{H}_{\rm code} \to \mathcal{H}_{\rm phys}$ with operator map $\varphi$, the following are equivalent:
\begin{enumerate}
        \item The operator map $\varphi$ is unital, i.e. $\varphi(I_{\rm code}) = I_{\rm phys}$.
        \item $\varphi(U)$ is unitary for every unitary $U$.
        \item The ratio $m = \dim \mathcal{H}_{\rm phys}/\dim \mathcal{H}_{\rm code}$ is an integer and there is a unitary $W:{\cal H}_{\rm code} \otimes \mathbb{C}^m \to \mathcal{H}_{\rm phys}$ such that 
        \[
        \varphi({\cal O}) = W ({\cal O} \otimes I_m ) W^\dagger \;, 
        \]
        and there is some unit vector  $e_1 \in \mathbb{C}^m$ such that 
        \[
        T |\psi \rangle = W (|\psi \rangle \otimes e_1 ) \qquad \forall |\psi \rangle \in {\cal H}_{\rm code} \;. 
        \]
        \item The algebra $\varphi(B(\mathcal{H}_{\rm code}))$ is a unital von Neumann factor subalgebra of $B(\mathcal{H}_{\rm phys})$.
    \end{enumerate}
\end{proposition}
\begin{proof}
Since $\varphi$ is a $*$-homomorphism, we have for any unitary $U \in B(\mathcal{H}_{\rm code})$:
\begin{equation}
\varphi(U) \varphi(U)^\dagger = \varphi(U)\varphi(U^\dagger) = \varphi(U U^\dagger) = \varphi(I_{\rm code}) \;.
\label{Uproof}
\end{equation}
Assuming 1, we thus have $\varphi(U) \varphi(U)^\dagger = I_{\rm phys}$ (and similarly $\varphi(U)^\dagger \varphi(U) = I_{\rm phys}$) so 2 holds. Assuming 2, the left side of (\ref{Uproof}) is $I_{\rm phys}$ for any $U$, so 1 holds.

Next, we show $1 \Leftrightarrow 3$. We have  $B(\mathcal{H}_{\rm code}) \equiv M_n(\mathbb{C})$ where $n = \dim \mathcal{H}_{\rm code}$. A standard result in the representation theory of matrix algebras is that every unital $*$-representation of $M_n(\mathbb{C})$ is a direct sum of copies of the defining representation (see, for example, \cite{Takesaki2002}). Thus, we must have that $\dim \mathcal{H}_{\rm phys} = m \dim \mathcal{H}_{\rm code}$ and in some basis, $\varphi({\cal O}) = {\cal O} \otimes I_m$. Letting this basis be the image of $\{|i\rangle \otimes |j\rangle\}$ under a unitary $W$ and defining $V = W^\dagger T$, we can write
\[
V |i \rangle = V_{i}^{jk} |j \rangle \otimes |k \rangle \equiv V^{(k)} |i \rangle \otimes |k \rangle\; .
\]
The relation $\varphi({\cal O}) T = T {\cal O}$ then gives
\[
{\cal O} V^{(k)} = V^{(k)} {\cal O}
\]
for any ${\cal O}$, so we can conclude that $V^{(k)} = c_k I_n$ for scalars $c_k$. Thus,
\[
V |\psi \rangle = |\psi \rangle \otimes \sum_k c_k |k \rangle \equiv |\psi \rangle \otimes e_1  \implies T |\psi \rangle = W (|\psi \rangle \otimes e_1)\; .
\]
Here, $e_1$ is a unit vector since $T$ is an isometry. This completes the proof that $1 \implies 3$. Conversely, assuming 3, we have that $\varphi(I_{\rm code}) = W (I_{\rm code} \otimes I_m) W^\dagger = I_{\rm phys}$, showing 1.

To show the equivalence between 3 and 4, we note that 3 implies $\varphi(B(\mathcal{H}_{\rm code})) = W(B(\mathcal{H}_{\rm code}) \otimes I_m)W^\dagger$ which is the standard form of a unital von Neumann factor subalgebra. Conversely, if $\varphi(B(\mathcal{H}_{\rm code}))$ is unital as a subalgebra of \(B(\mathcal H_{\rm phys})\), there is some \(\mathcal O\in B(\mathcal H_{\rm code})\) with \(\varphi(\mathcal O)=I_{\rm phys}\). But \(\varphi(I_{\rm code})\) is the
identity of the image algebra $\varphi(B(\mathcal{H}_{\rm code}))$, so
\[
\varphi(I_{\rm code})=I_{\rm phys}~.
\]
Thus \(4\Rightarrow 1\), and the implication \(1\Rightarrow 3\) proved above
gives \(4\Rightarrow 3\).
\end{proof}

Finally, we can define the factor version of a code network. We call this a {\it factor code network}. 
\begin{definition}
    A {\it factor code network} is a code network $\{\mathcal{H}_i\},\{T_{ji}\}$ of quantum codes together with a corresponding set of unital operator maps $\varphi_{ji}$ for $i \le j$ satisfying
    \[
    \varphi_{kj} \circ \varphi_{ji} = \varphi_{ki}  \qquad i \le j \le k \;,
    \]
    and
    \[
    \varphi_{ji}({\cal O}) T_{ji} = T_{ji} {\cal O}  \qquad \forall {\cal O} \in B(\mathcal{H}_i)\; .
    \]
\end{definition}
An immediate result is: 
\begin{proposition}
    In a factor code network, the algebras $\mathcal{M}_i^{(j)} \equiv \varphi_{ji}(B(\mathcal{H}_i))$ are unital von Neumann factor subalgebras of $B(\mathcal{H}_j)$ with the property that $\mathcal{M}_i^{(k)} \subset \mathcal{M}_j^{(k)}$ for $i < j < k$.
\end{proposition}
\begin{proof}
    For $i < j$, $T_{ji} : \mathcal{H}_i \to \mathcal{H}_j$ is a factor code, so by the previous proposition, $\mathcal{M}_i^{(j)}$ is a unital von Neumann factor subalgebra of $B(\mathcal{H}_j)$. Any operator in $\mathcal{M}_i^{(k)}$ is of the form $\varphi_{ki}({\cal O}) = \varphi_{kj} (\varphi_{ji}({\cal O}))$ so it is also in $\mathcal{M}_j^{(k)}$. Thus, $\mathcal{M}_i^{(k)} \subset \mathcal{M}_j^{(k)}$ for $i < j < k$.
\end{proof}
For a factor code network corresponding to a partially ordered set with a largest element $0$, we thus have a partially-ordered set of von Neumann factor subalgebras of $B(\mathcal{H}_0)$ whose inclusion structure is dictated by the ordering.

\subsubsection*{Central question}

The central question that we will address in this note is when a code network can be promoted to a factor code network. In cases where this promotion is possible, we would like to understand how much freedom we have in specifying this promotion, and what additional properties the resulting family of subalgebras has. For example:
\begin{itemize}
    \item Given two subalgebras $\mathcal{M}_i^{(k)}$ and $\mathcal{M}_j^{(k)}$, when do they generate the full algebra $\mathcal{M}_k^{(k)}$? 
    \item Given two subalgebras $\mathcal{M}_i^{(k)}$ and $\mathcal{M}_j^{(k)}$, when do they commute? 
    \item Given two subalgebras $\mathcal{M}_i^{(k)}$ and $\mathcal{M}_j^{(k)}$, when do they intersect only on multiples of the identity? 
    \item For algebras $\mathcal{M}_i^{(k)} \subset \mathcal{M}_j^{(k)}, \mathcal{M}_{\hat{j}}^{(k)} \subset \mathcal{M}_l^{(k)}$, when do these form a nondegenerate commuting square?
\end{itemize}

\section{Promoting codes to factor codes}

\subsection{Factor code families}

Consider a quantum code $T : \mathcal{H}_{\rm code} \to \mathcal{H}_{\rm phys}$. For a given operator map $\varphi$, the $*$-homomorphism property ensures that $\varphi$ preserves the $*$-algebra structure such that $\varphi(B(\mathcal{H}_{\rm code}))$ forms a $*$-subalgebra of $\mathcal{H}_{\rm phys}$. According to Proposition \ref{factor}, if the map $\varphi$ is unital so that $(T, \varphi)$ gives a factor code, the ratio $\dim  \mathcal{H}_{\rm phys} / \dim \mathcal{H}_{\rm code}$ must be an integer. We will now see that if this dimension condition is satisfied, there is a family of such unital operator maps related via conjugation by unitary operators acting in the space transverse to $T \mathcal{H}_{\rm code}$. 

We will use that the construction of operator maps can be described in terms of a set of auxiliary isometries:
\begin{definition}
    Given an isometry $T : \mathcal{H}_{\rm code} \to \mathcal{H}_{\rm phys}$, with integer $m = \dim  \mathcal{H}_{\rm phys} / \dim \mathcal{H}_{\rm code}$, we define an {\it orthogonal isometric completion} of $T$ to be a set  $\{T^{(a)}, a=1\dots m\}$ of isometries where $T^{(1)} = T$ and $T^{(a)} \mathcal{H}_{\rm code}$ are orthogonal subspaces of $\mathcal{H}_{\rm phys}$.
\end{definition}
Associated to such a choice of auxiliary isometries, we can define an operator map
\begin{equation}
\label{defmap}
    \varphi({\cal O}) = \sum_a T^{(a)} {\cal O} (T^{(a)})^\dagger \; ,
\end{equation}
that we will see is a unital $*$-homomorphism. More generally, we have:
\begin{proposition}
\label{Tphi}
For a quantum code $T : \mathcal{H}_{\rm code} \to \mathcal{H}_{\rm phys}$, with integer $m = \dim  \mathcal{H}_{\rm phys} / \dim \mathcal{H}_{\rm code}$, there is a one-to-one correspondence between 
\begin{enumerate}
    \item unital operator maps, i.e. unital $*$-homomorphisms
$\varphi:B(\mathcal H_{\rm code}) \to B(\mathcal H_{\rm phys})$ satisfying
\[
\varphi(\mathcal O)T=T\mathcal O \qquad
\text{for all } \mathcal O\in B(\mathcal H_{\rm code}) \; .
\]
\item equivalence classes of orthogonal isometric completions
$\{T^{(a)}\}_{a=1}^m$ of $T$ where two completions are said to be equivalent if
\[
T_2^{(a)} = \sum_{b=2}^m U_{ab} T_1^{(b)}\;, \qquad \qquad a = 2, \dots , m\;, \qquad U \in U(m-1)\; .
\]
\item unital von Neumann factor subalgebras $\mathcal M\subset B(\mathcal H_{\rm phys})$ such that $\mathcal M\cong B(\mathcal H_{\rm code})$ and the code subspace \(T\mathcal H_{\rm code}\) is invariant under $\mathcal M$.
\end{enumerate}
After choosing an arbitrary fixed reference operator map $\varphi_0$, the above spaces can be identified with the quotient $U(nm - n)/U(m-1)$ where $n = \dim \mathcal{H}_{\rm code}$. 
\end{proposition}
\begin{proof}
    First, consider an orthogonal isometric completion $\{T^{(a)}\}$ and define the associated map (\ref{defmap}). This is a homomorphism, since we have
    \[
    (T^{(a)})^\dagger T^{(b)} =  \delta_{ab} I_{\rm code}
    \]
    by the assumed orthogonality property and so
\[
\sum_a T^{(a)} {\cal O}_1 (T^{(a)})^\dagger \sum_b T^{(b)} {\cal O}_2 (T^{(b)})^\dagger = \sum_a T^{(a)} {\cal O}_1 {\cal O}_2 (T^{(a)})^\dagger \; .
\]
 We also have that $\varphi({\cal O}^\dagger) = \varphi({\cal O})^\dagger$ by the definition and $\varphi(I_{\rm code}) =  \sum_a T^{(a)} (T^{(a)})^\dagger =I_{\rm phys}$ since this is a sum of projectors $P^{(a)} = T^{(a)} (T^{(a)})^\dagger$ to a complete set of orthogonal subspaces. Thus, $\varphi$ is a unital $*$-homomorphism. Finally, we have 
 \[
 \varphi({\cal O}) T = \sum_a T^{(a)} {\cal O} (T^{(a)})^\dagger T^{(1)} = T^{(1)} {\cal O} = T {\cal O} \; .
 \]
 It follows immediately from the definition that two orthogonal isometric completions in the same equivalence class give the same operator map. Conversely, given a unital operator map, Proposition \ref{factor} gives that there exists a unitary $W : \mathcal{H}_{\rm code} \otimes \mathbb{C}^m \to \mathcal{H}_{\rm phys}$ such that for some vector $e_1 \in \mathbb{C}^m$ we have $W(|\psi \rangle \otimes e_1) = T |\psi \rangle$ for all $|\psi \rangle \in \mathcal{H}_{\rm code}$ and $\varphi({\cal O}) = W({\cal O} \otimes I_m)W^\dagger$. Choosing an orthonormal basis $e_a, a=2,\dots,m$ for $\mathbb{C}^m \ominus \mathbb{C} e_1$, we can define auxiliary isometries
    \begin{equation}
    \label{Tdef}
    T^{(a)} |i\rangle = W (|i\rangle \otimes e_a) \; ,
    \end{equation}
    which define an operator map via (\ref{defmap}). 
We can show that different choices for $W$ and for the basis elements  give orthogonal isometric completions that are in the same equivalence class. First, from (\ref{Tdef}), we see that the unitary rotations of basis elements $\{e_a\},a=2, \dots , m$ lead to unitary rotations of $T^{(a)}$ that define the equivalence class. Further, for a fixed basis, if $W_1, W_2$ correspond to the same operator map, we have for any ${\cal O}$
\[
W_1({\cal O} \otimes I_m) W_1^\dagger = W_2({\cal O} \otimes I_m) W_2^\dagger \Leftrightarrow W_2^\dagger W_1({\cal O} \otimes I_m) = ({\cal O} \otimes I_m) (W_2^\dagger W_1)
\]
so $W_2^\dagger W_1$ commutes with ${\cal O} \otimes I_m$ for any ${\cal O}$ and must take the form $I_{\rm code} \otimes U$. Since both $W_1$ and $W_2$ take $\mathcal{H}_{\rm code}\otimes e_1$ to $T \mathcal{H}_{\rm code}$, we must have $U e_1 = e_1$, so $U$ corresponds to a unitary rotation among the basis elements $\{e_a\},a=2, \dots , m$, which again move between different orthogonal isometric completions in the same equivalence class. Thus, we have described a well-defined map from an operator map to an equivalence class of orthogonal isometric completions. 

Starting from an orthogonal isometric completion constructed in this way from a unital operator map $\varphi$, the construction (\ref{defmap}) gives back the same $\varphi$ since
\[
    \sum_a T^{(a)} {\cal O} (T^{(a)})^\dagger = \sum_{ij} {\cal O}_{ij} \sum_a T^{(a)} |i\rangle \langle j| (T^{(a)})^\dagger = \sum_{ij} {\cal O}_{ij} \sum_a W (|i\rangle \otimes e_a) (\langle j| \otimes e_a^\dagger) W^\dagger =  W ({\cal O} \otimes I_m) W^\dagger = \varphi({\cal O}) \; .
    \]
    Thus, we have established a one-to-one correspondence between 1 and 2. 

  The identification with $U(nm-n)/U(m-1)$ can be established by noting that for any reference $\varphi({\cal O})$ and corresponding equivalence class $S = \{\{T^{(a)}\}|\sum_a T^{(a)} {\cal O} (T^{(a)})^\dagger = \varphi({\cal O})\}$ of orthogonal isometric completions, any other orthogonal isometric completion can be obtained starting with an arbitrary representative $\{T^{(a)}_0\} \in S$ and taking 
  \[
  T^{(a)} = \hat{U} T^{(a)}_0 \qquad \qquad 
  \]
  for some unitary $\hat{U}$ on $\mathcal{H}_{\rm phys}$ acting as the identity on  $T \mathcal{H}_{\rm code}$ and as $U(T \mathcal{H}_{\rm code}^\perp) \equiv  U(nm - n)$ on the orthogonal complement. The equivalence relation gives the identification $U \equiv U (I_{\rm code} \otimes u)$ where $u$ is a unitary on $\mathbb{C}^{m-1}$ so we have one-to-one correspondence between distinct $\varphi$ elements of $Q = U(nm-n)/U(m-1)$ for a given $\varphi_0$.

  Finally, we show the equivalence between 1 and 3. Given a unital operator map $\varphi$,  Proposition \ref{factor} gives that $\varphi(B({\cal H}_{\rm code}))$ is a unital von Neumann factor subalgebra of $B({\cal H}_{\rm phys})$. This is isomorphic to $B({\cal H}_{\rm code})$ by definition and also leaves $T {\cal H}_{\rm code}$ invariant since $\varphi({\cal O}) T |\psi \rangle =  T {\cal O}|\psi \rangle \in T {\cal H}_{\rm code}$.  Conversely, consider any unital von Neumann factor subalgebra $\mathcal{M}$ isomorphic to $B({\cal H}_{\rm code})$ and leaving $T \mathcal{H}_{\rm code}$ invariant. We will show that the map
   \begin{equation}
   \label{map}   
   \kappa: \mathcal{M} \to B(\mathcal{H}_{\rm code}) \; , \qquad  \kappa({\cal O}) =  T^\dagger {\cal O} T \; 
   \end{equation}
   is a unital $*$-homomorphism. We note that for $A \in \mathcal{M}$, $A T$ is in $T \mathcal{H}_{\rm code}$ so for the projector $P = T T^\dagger$ to $T \mathcal{H}_{\rm code}$
   \[
   A T = P A T = T T^\dagger A T = T \kappa(A) \; .
   \]
   Thus, 
   \[
   \kappa(AB) = T^\dagger AB T = T^\dagger A T \kappa(B) = \kappa(A) \kappa(B) \; .
   \]
   We also have 
   \[
   \kappa(I_{\rm phys}) = T^\dagger T = I_{\rm code} \qquad \qquad \kappa(A^\dagger) = T^\dagger A^\dagger T = (T^\dagger A T)^\dagger = (\kappa(A))^\dagger \; ,
   \]
   thus, $\kappa$ is a unital $*$-homomorphism. 
   If $\kappa(A) = 0$ then also $\kappa(AB) = \kappa(A) \kappa(B) = \kappa(BA) = 0$ for any $B \in \mathcal{M}$. Thus, the kernel of $\kappa$ is a two-sided ideal (i.e. is preserved by left and right multiplication by elements of $\mathcal{M}$). Since $\mathcal{M} \equiv M_n(\mathbb{C})$ and $M_n(\mathbb{C})$ is simple, the only such ideals are $0$ and $M_n(\mathbb{C})$, but the latter is not possible since $\kappa$ is unital. Thus, $\ker \kappa = 0$. Since $\kappa$ maps between algebras of the same dimension, $\kappa$ is a $*$-isomorphism. Defining $\varphi$ to be the inverse of $\kappa$, we can check that the properties of $\kappa$ imply that $\varphi$ is a unital $*$-homomorphism with $\varphi({\cal O}) T = T {\cal O}$:
   \begin{eqnarray*}
        \kappa(\varphi({\cal O}_1) \varphi({\cal O}_2)) =  \kappa(\varphi({\cal O}_1)) \kappa(\varphi({\cal O}_2)) = {\cal O}_1{\cal O}_2 &\implies& \varphi({\cal O}_1) \varphi({\cal O}_2) = \varphi({\cal O}_1{\cal O}_2) \cr
        \kappa(I_{\rm phys}) = I_{\rm code} & \implies& \varphi(I_{\rm code}) = I_{\rm phys} \cr
        \kappa(\varphi({\cal O})^\dagger) = \kappa(\varphi({\cal O}))^\dagger = {\cal O}^\dagger &\implies&  \varphi({\cal O}^\dagger) = (\varphi({\cal O}))^\dagger \cr
        A T = T \kappa(A) &\implies& \varphi({\cal O}) T = T \kappa(\varphi({\cal O})) = T {\cal O} \; .
   \end{eqnarray*}
  We have thus established a one-to-one correspondence between 1 and 3.
    
    \end{proof}
The identification with $U(nm - n)/U(m-1)$ provide a natural measure on the space of possible unital operator maps, induced by the Haar measure on $U(nm - n)/U(m-1)$. We will use this measure below when describing properties that hold ``generically'' or when choosing a ``random'' unital operator map.

\section{Factor codes for tree networks}

In the previous section, we have seen that given a quantum code $T:\mathcal{H}_{\rm code} \to \mathcal{H}_{\rm phys}$, there is a family of factor codes that we can upgrade it to, related by a unitary equivalence. 

Next, we consider the situation where we have a code network associated with some Hasse diagram, and ask whether it is possible to define operator maps associated with each of the isometries in order to define a factor code network. We focus in this section on the special case where the code network is a tree (as in Figure \ref{fig:Hasse} (left)), showing that this upgrade is always possible.

\begin{proposition}
Consider a code network $\{\mathcal{H}_i\}_{i \in S}$ for which the Hasse diagram of the index set $S$ is a tree.
If ${\rm dim}(\mathcal{H}_i)$ divides ${\rm dim}(\mathcal{H}_j)$ for each $i < j$ then we can define unital operator maps $\varphi_{ji} : B(\mathcal{H}_i) \to B(\mathcal{H}_j)$ to promote the system to a factor code network.
\label{tree}
\end{proposition}
\begin{proof}
For each edge $i \to j$ in the Hasse diagram, we have Hilbert spaces $\mathcal{H}_i$ and $\mathcal{H}_j$ with integer $m_{ij} = {\rm dim} \mathcal{H}_j / {\rm dim} \mathcal{H}_i$. We can construct isometries $T^{(a)}_{ji} : \mathcal{H}_i \to \mathcal{H}_j$, $a = 1, \dots, m_{ij}$ with $T^{(1)}_{ji} = T_{ji}$ and the remaining isometries mapping to orthogonal subspaces of $\mathcal{H}_j$. For any $i < j$ with intermediate vertices, we have a unique path $i \to k_1 \to \cdots \to k_n \to j$ in the Hasse diagram, since we are assuming this is a tree. We can define $m_{ij} = m_{i k_1} m_{k_1 k_2} \cdots m_{k_n j}$ isometries 
\begin{equation} \label{eq:compositions}
T^{(a_1 \cdots a_{n+1})}_{ji} = T^{(a_1)}_{j k_n} T^{(a_2)}_{k_n k_{n-1}} \cdots T^{(a_{n+1})}_{k_1 i} \; .
\end{equation}
Again, these map to orthogonal subspaces. Finally, for any $i < j$ we define the operator map
\[
\varphi_{ji}({\cal O}_i)= \sum_a T^{(a)}_{ji} {\cal O}_i (T^{(a)}_{ji})^\dagger \; .
\]
According to Proposition \ref{Tphi}, since $\{T_{ji}^{(a)}\}$ defines an orthogonal isometric completion of $T_{ji}$, $\varphi_{ji}$ defines a unital $*$-homomorphism satisfying 
\[
\varphi_{ji}({\cal O}) T_{ji} = T_{ji} {\cal O}  \qquad \forall {\cal O} \in B(\mathcal{H}_i) \quad   i \le j \; .
\]
The remaining composition property for the operator maps follows immediately from \eqref{eq:compositions}. For $i \le j \le k$, we have
\[
\varphi_{kj} (\varphi_{ji}({\cal O}_i)) = \sum_{a,b}  T^{(a)}_{kj} T^{(b)}_{ji} {\cal O}_i (T^{(b)}_{ji})^\dagger (T^{(a)}_{kj})^\dagger = \sum_{a,b}  T^{(a,b)}_{ki} {\cal O}_i (T^{(a,b)}_{ki})^\dagger = \varphi_{ki} ({\cal O}_i) \qquad i \le j \le k \; .
\]
\end{proof}

\subsection{Duplex codes}

For a single quantum code, or for code network where the Hasse diagram is a simple chain of Hilbert spaces, the structure is entirely determined by the dimensions of the Hilbert spaces. Any two such code networks or associated factor code networks with the same set of Hilbert space dimensions are related by unitary equivalence. We have seen this explicitly in Proposition \ref{Tphi} for a single factor code.

On the other hand, when we have a network where two different code Hilbert spaces are mapping into the same physical Hilbert space, there is additional invariant (basis-independent) information in the code network, namely the relative position of the two code subspaces in the physical Hilbert space. Similarly, there are various inequivalent completions to factor codes, and the subalgebras associated with these different completions can have different properties with respect to commutation, intersection and the algenbra they generate (join). In this section, we will explore and characterize this additional structure.

We focus on the simple case of a single physical Hilbert space in which two different logical Hilbert spaces are encoded:
\begin{definition}
A {\it duplex code} is a triple of Hilbert spaces $(\mathcal{H}_A,\mathcal{H}_B,\mathcal{H})$ with isometries $T_A:\mathcal{H}_A \to \mathcal{H}$ and $T_B:\mathcal{H}_B \to \mathcal{H}$. We define the associated code subspaces
\[
A \equiv T_A \mathcal{H}_A~, \qquad B \equiv T_B \mathcal{H}_B
\]
and projectors
\[
P_A = T_A T_A^\dagger~, \qquad P_B = T_B T_B^\dagger \; .
\] 
\end{definition}

\subsubsection*{Principal angles between code subspaces}

The invariant geometric data describing the relative orientation of the two code subspaces $A$ and $B$ is a set of {\it principal angles} $\theta_i$ whose number is equal to the dimension of the smaller Hilbert space.\footnote{Principal angles and the associated canonical decomposition of a pair of subspaces are classical tools in operator theory and numerical linear algebra \cite{Halmos1969,BjorckGolub1973,KnyazevArgentati2002}.} One way to visualize these is via a decomposition of the full Hilbert space into a direct sum with the following orthogonal parts:
\begin{itemize}
    \item The intersection $A \cap B$. 
    \item A set of two-dimensional spaces $C_i$ such that $A \cap C_i$ and $B \cap C_i$ are one-dimensional subspaces with relative angle $0 < \theta_i < \pi/2$.
    \item The intersection $A \cap B^\perp$.
    \item The intersection $A^\perp \cap B$.
    \item The orthogonal space $A^\perp \cap B^\perp$.
\end{itemize}
Along with the principal angles $0 < \theta_i < \pi/2$ associated with the subspaces $C_i$, we define a number of principal angles $\theta_i = 0$ equal to the dimension of the intersection $A \cap B$ and a number of principal angles $\theta_i = \pi/2$ equal to the minimum of $\dim A \cap B^\perp$ and $\dim A^\perp \cap B$.

The principal angles can be efficiently calculated by noting that $\{\cos \theta_i \}$ are the singular values of the operator $P_{B}|_A$ if $\dim A < \dim B$ or $P_{A}|_B$ of $\dim B < \dim A$. Alternatively, the principal angles with $\theta_i < \pi/2$ are related to the nonzero eigenvalues of $P_A P_B P_A$ (or equivalently of $P_B P_A P_B$) by
$\lambda_i = \cos^2 \theta_i$.\footnote{For example, with two one-dimensional Hilbert spaces embedded in a higher dimensional space, we have $P_A = |v \rangle \langle v|$ and $P_B = |w \rangle \langle w|$ so $P_A P_B P_A = \cos^2(\theta) |v \rangle \langle v|$ where $\cos^2(\theta) = |\langle v |w \rangle|^2$, so we have a single angle characterizing the relative orientation, as expected.}

These principal angles will play an important role in some of our discussions below characterizing the possible factor code networks arising from networks that include duplex codes.

\subsubsection*{Properties of $\mathcal{M}_A$ and $\mathcal{M}_B$}

Given a duplex code with any nonzero principal angles between the code subspaces $A$ and $B$, there will be various inequivalent factor code networks that we can associate. When choosing the auxiliary isometries $T_A^{(a)}$ and $T_B^{(a)}$ to define unital operator maps, the subspaces $T_A^{(a)} \mathcal{H}_A$ now each have a set of principal angles with $B$ and the subspaces $T_B^{(a)} \mathcal{H}_B$ each have a set of principal angles with $A$. These various principal angles can differ in different factor code networks associated with the same underlying duplex code.

In some cases, it is possible to use this freedom to demand additional special properties of the factor code network. We will focus on properties of the algebras $\mathcal{M}_A = \varphi_A(B(\mathcal{H}_A))$ and $\mathcal{M}_B = \varphi_B(B(\mathcal{H}_A))$. The following propositions characterize when each of the following properties are possible and/or generic.
\begin{itemize}
    \item The algebra generated by $\mathcal{M}_A$ and $\mathcal{M}_B$ is the full algebra $B(\mathcal{H})$.
    \item The intersection of $\mathcal{M}_A$ and $\mathcal{M}_B$ is as small as possible, i.e. $\mathcal{M}_A \cap \mathcal{M}_B = \mathbb{C} I_{\rm phys}$.
    \item The algebras $\mathcal{M}_A$ and $\mathcal{M}_B$ commute with each other.
\end{itemize}

\subsubsection{Algebra generated by $\mathcal{M}_A$ and $\mathcal{M}_B$}

We begin with:
\begin{proposition}
\label{Union1}
    Suppose we have Hilbert spaces ${\cal H}_A$, ${\cal H}_B$, and ${\cal H}$ with isometries $T_A : {\cal H}_A \to {\cal H}$ and  $T_B : {\cal H}_B \to {\cal H}$. 
    In the space of possible extensions of the isometries to unital operator maps $\varphi_A$, $\varphi_B$ satisfying the properties of Proposition \ref{tree}, a Haar random choice will have the property that the subalgebras $\mathcal{M}_A = \varphi_A(B(\mathcal{H}_A))$ and $\mathcal{M}_B= \varphi_B(B(\mathcal{H}_B))$ together generate $B(\mathcal{H})$ unless \begin{itemize}
        \item $\dim \mathcal{H}_A = 1$ and $  \dim \mathcal{H}_B < \dim \mathcal{H}$~.
        \item $\dim \mathcal{H}_B = 1$ and $ \dim \mathcal{H}_A < \dim \mathcal{H}$~.
        \item
        $T_A \mathcal{H}_A = T_B \mathcal{H}_B \ne {\cal H}$~.
        \item
        $T_A \mathcal{H}_A$ is the orthogonal complement of $T_B \mathcal{H}_B$~.
    \end{itemize}
     In all of these cases, no choices for $\mathcal{M}_A$ and $\mathcal{M}_B$ generate $B(\mathcal{H})$.    
\end{proposition}
\begin{proof}
For $\dim \mathcal{H}_A = 1$ or $\dim \mathcal{H}_B = 1$, the algebra associated with the one-dimensional Hilbert space is scalar, so the two algebras together generate $B(\mathcal{H})$ if and only if the other algebra is already equal to $B(\mathcal{H})$, that is,  if the larger subspace is all of $\mathcal{H}$.

Otherwise, let $a > 1$ and $b > 1$ be the dimensions of $\mathcal{H}_A$ and $\mathcal{H}_B$, and let the dimension of $\mathcal{H}$ be $n = k a = m b$.
Let $A \equiv T_A \mathcal{H}_A \subset \mathcal{H}$ and $B \equiv T_B \mathcal{H}_B \subset \mathcal{H}$. Then ${\cal H}$ is unitarily equivalent to $A \otimes \mathbb{C}^k$ and to $B \otimes \mathbb{C}^m$. For each choice of operator maps, the algebra $\mathcal{M}_A$ can be represented as
\[
\mathcal{M}_A = V(B(A) \otimes I_k)V^\dagger~,
\]
where $V: A \otimes \mathbb{C}^k \to \mathcal{H}$ is a unitary satisfying $V( |\alpha\rangle \otimes e_1) = |\alpha\rangle$ (equivalent to a general unitary from $A \otimes \mathbb{C}^{k-1}$ to $A^\perp$). 
Similarly, the algebra $\mathcal{M}_B$ can be represented as
\[
\mathcal{M}_B = W(B(B) \otimes I_m)W^\dagger~,
\]
where $W: B \otimes \mathbb{C}^m \to \mathcal{H}$ is a unitary satisfying $W(|\beta\rangle \otimes e_1) = |\beta\rangle$.
The commutant algebras (within the algebra $B(\mathcal{H})$) are then given by
\[
\mathcal{M}_A' = V(I_a \otimes M_k(\mathbb{C}))V^\dagger~, \quad \mathcal{M}_B' = W(I_b \otimes M_m(\mathbb{C}))W^\dagger \; .
\]
If $\mathcal{M}_A$ and $\mathcal{M}_B$ together generate the full algebra on $\mathcal{H}$, then $\mathcal{M}_A' \cap \mathcal{M}_B' = (\mathcal{M}_A \vee \mathcal{M}_B)' = \mathbb{C} I_\mathcal{H}$. Conversely, if $\mathcal{M}_A' \cap \mathcal{M}_B' = \mathbb{C} I_\mathcal{H}$, then $\mathcal{M}_A \vee \mathcal{M}_B = (\mathcal{M}_A \vee \mathcal{M}_B)'' = (\mathcal{M}'_A \cap \mathcal{M}_B')' = (\mathbb{C} I_\mathcal{H})' = B(\mathcal{H})$. 
Thus, the statement that $\mathcal{M}_A$ and $\mathcal{M}_B$ generate $B(\mathcal{H})$ is equivalent to the statement that $$
\mathcal{M}'_A \cap \mathcal{M}_B' = \mathbb{C} I_\mathcal{H}~.
$$

We will argue that the latter is true for typical (Haar-random) choices of $V$ and $W$. First, we note that $\mathcal{M}'_A \cap \mathcal{M}_B' \ne \mathbb{C} I_\mathcal{H}$ if and only if the intersection contains some projection $P_R$ of rank $0 < R < n$. Since $P_R \in \mathcal{M}'_A$, we can write it as $V (I_a \otimes P_A) V^\dagger$ for some projection $P_A$, and since $P_R \in \mathcal{M}'_B$, we can also write it as $W (I_b \otimes P_B) W^\dagger$ for some projection $P_B$. This shows that the rank $R$ is divisible by both $a$ and $b$. 

Let $C_V$ be the space of operators of the form $V(I_a \otimes M_k(\mathbb{C}))V^\dagger$ for $V$ satisfying the conditions above and $C_W$ be the space of operators of the form $W (I_b \otimes M_k(\mathbb{C})) W^\dagger$ for $W$ as above. The set $S_V$ of allowed $V$s is isomorphic to $U(n-a)$ since $V$ is a unitary on an $n$-dimensional space with a fixed action on an $a$-dimensional subspace. We thus have real dimension $\dim S_V = (n-a)^2$. Similarly, the set $S_W$ of allowed $W$s is isomorphic to $U(n-b)$ and has real dimension $\dim S_B = (n-b)^2$.

In the space of operators on $\mathcal{H}$, let $\mathcal{P}_A$ be the set of all projectors of rank $R$ that appear in some $C_V$, $\mathcal{P}_B$ be the set of all projectors of rank $R$ that appear in some $C_W$, and $\Sigma \equiv \mathcal{P}_A \cap \mathcal{P}_B$ be the intersection.

Consider the space of triples 
\[
Y \equiv \{(V,W,P_R) \,| \, P_R \in C_V \cap C_W \} \; .
\]
This is a subset of $S_V \times S_W \times \Sigma_R$. We will show that the image $Y \to S_V \times S_W$ is a measure zero set and thus that generic $\mathcal{M}_A'$ and $\mathcal{M}_B'$ intersect only on multiples of the identity.

We can think of $Y$ as having a base $\Sigma$ with fibers $F_A(P) \equiv \{ V \in S_V \,|\, P \in C_V\}$ and $F_B(P) \equiv \{ W \in S_W \, | \, P \in C_W\}$ in $S_V$ and $S_W$ respectively. We will see that the possible $P$s in $\Sigma$ can be divided into different classes (related to how the subspace $P \mathcal{H}$ sits relative to $A$ and $B$). The dimensions of the fibers and the relevant portion of $\Sigma$ depends on the class. For the portion of $Y_R$ associated with each class of projector $P^{(i)}$, we have 
\begin{equation}
\dim Y^{(i)} \le \dim \Sigma^{(i)} +  \dim F_A(P^{(i)}) +  \dim F_B(P^{(i)})~.
\label{E1}
\end{equation}
We will show that the right-hand side is smaller than $\dim (S_V \times S_W)$ apart from the exceptional cases.

We will first understand how projectors can be divided into different classes (i.e. we will understand the so-called ``stratification'' of $\Sigma$). 
We recall that the space $\mathcal{P}_A$ describes operators of the form
\[
P = V (I_a \otimes P_A) V^\dagger~,
\]
for some projector $P_A$ on $\mathbb{C}^k$ and some $V \in S_V$. We will show that this space 
splits up into a one-dimensional space of $S_V-$orbits labeled by the angle $\theta$ between the distinguished vector $e_1 \in \mathbb{C}^k$ associated with $A$ and the subspace $P_A \mathbb{C}^k$.

Consider any specific projector $P = V_0 (I_a \otimes P_A) V_0^\dagger$. We will use $V_0 : A \otimes \mathbb{C}^k \to  \mathcal{H}$ to identify $A \otimes \mathbb{C}^k \equiv \mathcal{H}$. In this case, using a basis $|\alpha_k\rangle \otimes e_i$ ordered as $|\alpha_1\rangle \otimes e_1, |\alpha_2\rangle \otimes e_1, \dots$, the matrix representation for $P$ takes a block form $P = I_a \otimes P_A$, and the allowed $V$ in $S_V$ take the form 
\[
V = \left(\begin{array}{cc}I_a & 0 \cr 0 & V_\perp \end{array}\right)~, \qquad V_\perp \in U(n-a) \; .
\]
We will show that by choosing some $V$ of this form with $V_\perp = I_a \otimes v$ and $v \in U(k-1)$, we can arrange that $V P V^\dagger = I_a \otimes P_\theta$ with
\begin{equation} \label{eq:Ptheta}
P_\theta =  \left(\begin{array}{cccc} \cos^2 \theta & \cos \theta \sin \theta & 0 & 0 \cr \cos \theta \sin \theta & \sin^2 \theta & 0 & 0 \cr 0 & 0 & I_{r-1} & 0 \cr 0 & 0 & 0 & 0_{k-r-1} \end{array} \right)~. 
\end{equation}
Indeed, taking the projector
\[
P_A =  \left(\begin{array}{cc} p & q^\dagger \cr q & M \end{array}\right) \; ,
\]
and $V$ of the form described above, we have $V (I_a \otimes P_A) V^\dagger = I_a \otimes \hat{P}$ where
\[
\hat{P} = \left(\begin{array}{cc} 1 & 0 \cr 0 & v \end{array}\right) \left(\begin{array}{cc} p & q^\dagger \cr q & M \end{array}\right) \left(\begin{array}{cc} 1 & 0 \cr 0 & v^\dagger \end{array}\right) = \left(\begin{array}{cc} p & q^\dagger v^\dagger \cr v q & vMv^\dagger \end{array}\right)~.
\]
Choosing $v$ to diagonalize $M$, defining $D = v M v^\dagger$ and defining $s \equiv vq$, the condition $\hat{P}^2 = \hat{P}$ gives
\[
p^2 + s^\dagger s = p~, \qquad (D + (p-1)I_{k-1}) s = 0~, \qquad s s^\dagger + D^2 = D \; .
\]
The first equation implies $p - p^2 \ge 0$, so $0 \le p \le 1$ and we can set $p = \cos^2 \theta$. Since $s s^\dagger$ has rank 1 with non-zero eigenvalue (according to the first equation) $p - p^2 = \cos^2 \theta \sin^2 \theta$, the last equation implies that $D$ has one eigenvalue $\lambda$ with $\lambda - \lambda^2 = \cos^2 \theta \sin^2 \theta$, and the remaining eigenvalues equal to 0 or 1. We can thus write (for some choice of $v$) $D = \text{diag}(\lambda, 1, \dots, 1, 0, \dots, 0)$. Since $s s^\dagger = D - D^2 = \text{diag}(\cos^2 \theta \sin^2 \theta, 0, \dots, 0)$, we conclude that the first component of $s$ has magnitude $\sin \theta \cos \theta$ and the rest are zero. Finally, the first component of the middle equation gives $\lambda = 1 - p = \sin^2 \theta$. We can take the first component of $s$ to be real and positive by a choice of the top left component of $v$. Thus, we have shown the $P_A$ can be placed in the claimed standard form depending only on a single parameter $\theta$.\footnote{This parameter has the interpretation as the angle between $e_1$ and the hyperplane associated with $P_A$.}

The classification above gives three classes of  $P \in \mathcal{P}_A$, depending on whether $\theta = 0$, $\theta = \pi/2$, or $0 < \theta < \pi/2$. Similarly for $P \in \mathcal{P}_B$, we have an equivalent classification, with another angle $\chi$. We will now bound the dimension of the fibers $F_A(P)$ and $F_B(P)$ for each case. 
Let $P = V_0(I_a \otimes P_\theta) V_0^\dagger$. The fiber $F_A(P)$ is the set of $V$s for which there is a $P_A$ satisfying $V(I_a \otimes P_A)V^\dagger = P$. Here, $P_A$ must be in the same class as $P_\theta$. Let $d_\theta(P)$ be the dimension of the space of all such $P_A$s. For a fixed choice of $(V,P_A)$ satisfying $V(I_a \otimes P_A)V^\dagger = P$, some other $\hat{V} \in S_V$ will give $\hat{V}(I_a \otimes P_A)\hat{V}^\dagger = P$ for the same $P_A$ if and only if $\hat{V} V^\dagger P (\hat{V} V^\dagger)^\dagger = P$, that is, $\hat{V} V^\dagger$ is in the \emph{stabilizer} of $P$. Defining $D_\theta(P)$ to be the dimension of this stabilizer, we thus have the upper bound:
\begin{equation}
\label{E2}
\dim F_A(P) \le d_\theta(P) + D_\theta(P) \; .
\end{equation}
The dimensions $d_\theta(P)$ and $D_\theta(P)$ depend only on which of the three classes $P$ is in. To compute $d_\theta$, we can use that it is the dimension of the space of matrices of the form
\[
\hat{P} = \left(\begin{array}{cc} 1 & 0 \cr 0 & v \end{array}\right) P_\theta \left(\begin{array}{cc} 1 & 0 \cr 0 & v^\dagger \end{array}\right)~, \qquad v \in U(k-1) \; .
\]
This is equal to the dimension $(k-1)^2$ of the unitary group $U(k-1)$ minus the dimension of the stabilizer of $P_\theta$ in $U(k-1)$. This stabilizer dimension is equal to the dimension of the space of $k\times k$ hermitian matrices of the form $\text{diag}(0,h)$ that commute with $P_\theta$. Using the explicit form \eqref{eq:Ptheta}, a straightforward computation gives:
\[
d_\theta(P) = \left\{ \begin{array}{cc} 2(r-1)(k-r)~, & \qquad \theta = 0~, \cr
2r(k-r-1)~, & \qquad \theta = \pi/2~, \cr
2r(k-r) - 1~, & \qquad \text{otherwise}~. \end{array} \right.
\]
Similarly, the stabilizer dimension $D_\theta(P)$ is the dimension of the space of $n \times n$ Hermitian matrices of the form $\text{diag}(0,H)$ with $H$ hermitian on $A^\perp$ that commute with $I_a \otimes P_\theta$.
An explicit computation shows that:
\[
D_\theta(P) = \left\{ \begin{array}{cc} (R-a)^2 + (n-R)^2~, & \qquad \theta = 0~, \cr
R^2 + (n-R-a)^2~, & \qquad \theta = \pi/2~, \cr
(R-a)^2 + (n - R - a)^2~, & \qquad \text{otherwise}~. \end{array} \right.
\]
We have a completely parallel calculation for the dimensions relevant to the fiber $F_B(P)$.

Let us now return to our dimension comparison. Each $P$ in $\Sigma$ has an angle $\theta$ associated with its standard form in ${\cal P}_A$ and an angle $\chi$ associated with its standard form in ${\cal P}_B$. We can stratify $\Sigma$ into nine parts each labelled with one of the three possibilities $\theta = 0,\theta = \pi/2,\theta \in (0,\pi/2)$ for $\theta$ and one of the three possibilities $\chi = 0,\chi = \pi/2, \chi \in (0,\pi/2)$ for $\chi$. Labelling these discrete parts of $\Sigma$ as $\Sigma^{(i)}$, and calling $Y^{(i)}$ the part of $Y$ over the base $\Sigma^{(i)}$, we wish to understand when
\begin{equation}
\label{Ycond}
\dim Y^{(i)} < \dim S_V + \dim S_W \; 
\end{equation}
starting from (\ref{E1}). We note that there is the following inequality:
\begin{equation}
\dim \Sigma^{(i)} \le \min(\dim {\cal P}^{(i)}_A,\dim {\cal P}^{(i)}_B) = \dim {\cal P}^{(i)}_A + \dim {\cal P}^{(i)}_B - \max(\dim {\cal P}^{(i)}_A,\dim {\cal P}^{(i)}_B) \; .
\label{E3}
\end{equation}
Furthermore, since for any $P \in {\cal P}_A$, we can write $P = V(I_a \otimes P_\theta)V^\dagger$, we have
\begin{equation}
\dim {\cal P}_A^i = \dim S_V - D_\theta^{(i)} + \dim S_\theta^{(i)} \; ,
\label{E4}
\end{equation}
where $\dim S_\theta^{(i)} = 0,1$ is the dimension of the set of $\theta$s corresponding to the $i$th stratum. Combining equations (\ref{E1}), (\ref{E2}), (\ref{E3}), and (\ref{E4}), we find that
\[
\dim Y^{i} \le [\dim S_V + \dim S_W] + [d_\theta^{(i)} + d_\chi^{(i)} + \dim S_\theta^{(i)} + \dim S_\chi^{(i)} - \max(\dim {\cal P}^{(i)}_A,\dim {\cal P}^{(i)}_B)]~.
\]
So the condition (\ref{Ycond}) will be satisfied provided that
\[
\max(\dim {\cal P}^{(i)}_A,\dim {\cal P}^{(i)}_B) > d_\theta^{(i)} + d_\chi^{(i)} + \dim S_\theta^{(i)} + \dim S_\chi^{(i)} \; ,
\]
which will hold if
\begin{equation}
{1 \over 2}(\dim {\cal P}^{(i)}_A+\dim {\cal P}^{(i)}_B) > d_\theta^{(i)} + d_\chi^{(i)} + \dim S_\theta^{(i)} + \dim S_\chi^{(i)} \; .
\label{final}
\end{equation}
Using the explicit results above, we have that
\[
{1 \over 2} \dim {\cal P}^{(i)}_A - d_\theta^{(i)} - \dim S_\theta^{(i)} = {1 \over 2} (\dim S_V - D_\theta^{(i)} - 2 d_\theta^{(i)} - \dim S_\theta^{(i)}) = \left\{ \begin{array}{cc} (r-1)(k-r)(a^2-2)~, & \qquad \theta = 0~, \cr
r(k-r-1)(a^2-2)~,& \qquad \theta = \pi/2~, \cr
(r(k-r)-1/2)(a^2-2) - 1/2~, & \qquad \text{otherwise}~. \end{array} \right.
\]
with similar results for the $B$ side. Since we have $a > 1$ and $0 < r < k$, the right side is non-negative in all cases and zero if and only if $(\theta = 0,r=1)$ or $(\theta = \pi/2, r= k-1)$. Looking at the projector \eqref{eq:Ptheta}, these correspond to $P {\cal H}=A$ and $P {\cal H} = A^{\perp}$ respectively. Taking into account the parallel results for ${\cal P}_B$, (\ref{final}) will hold in all cases unless we have both $P {\cal H} \in \{A, A^\perp\}$ and $P {\cal H} \in \{B, B^\perp\}$. 
 Since $A$  and $A^\perp$ are always fixed by ${\cal M}_A$, in the cases $A = B$ and $A=B^\perp$, $A$ will be fixed by both $\mathcal{M}_A$ and $\mathcal{M}_B$ and the algebra they generate, so this cannot be all of $B({\cal H})$ unless $A=B={\cal H}$. Apart from these special cases, every stratum of $Y$ has dimension strictly smaller than $\dim S_V + \dim S_W$, so the projection of $Y$ to $S_V \times S_W$ has measure zero. This completes the proof.
\end{proof}

Via Proposition \ref{Tphi}, an equivalent version of the proposition is the following:
\begin{corollary}
\label{Cunion}
    Let A and B be $a$ and $b$-dimensional subspaces of an $n$-dimensional Hilbert space $\mathcal{H}$, such that $n = k a = m b$ for integers $k,m$. Let $F_A$ be the set of von Neumann factor subalgebras $\mathcal{M}_A \subset B(\mathcal{H})$ with $\mathcal{M}_A \equiv M_a(\mathbb{C})$ and $\mathcal{M}_A A = A$. Let $F_B$ be the set of von Neumann factor subalgebras $\mathcal{M}_B \subset B(\mathcal{H})$ with $\mathcal{M}_B \equiv M_b(\mathbb{C})$ and $\mathcal{M}_B B = B$. Then away from a measure zero set in $F_A \times F_B$, $\mathcal{M}_A$ and $\mathcal{M}_B$ together generate $B(\mathcal{H})$ unless
    \begin{itemize}
        \item $\dim A = 1 $ and $ \dim B < \dim \mathcal{H}$~.
        \item $\dim B = 1 $ and $ \dim A < \dim \mathcal{H}$~.
        \item
        $A = B \ne {\cal H}$~.
        \item
        $A = B^\perp$~.
    \end{itemize}
\end{corollary}

\subsubsection{The intersection of $\mathcal{M}_A$ and $\mathcal{M}_B$}

Next, we ask when the intersection of $\mathcal{M}_A$ and $\mathcal{M}_B$ is as small as possible:
\begin{proposition}
\label{prop:intersection}
    Suppose we have Hilbert spaces ${\cal H}_A$, ${\cal H}_B$, and ${\cal H}$ with isometries $T_A : {\cal H}_A \to {\cal H}$ and  $T_B : {\cal H}_B \to {\cal H}$. 
    In the space of possible extensions of the isometries to unital operator maps $\varphi_A$, $\varphi_B$ satisfying the properties of Proposition \ref{tree}, a Haar random choice will have the property that $\mathcal{M}_A = \varphi_A(B(\mathcal{H}_A))$ and $\mathcal{M}_B = \varphi_B(B(\mathcal{H}_B))$ intersect only on multiples of the identity unless one of following holds:
    \begin{itemize}
    \item $A = \mathcal{H}$ and $ \dim B  > 1$~.
    \item $B = \mathcal{H}$ and $ \dim A  > 1$~.
    \item $\dim A = \dim B = \dim \mathcal{H}/2 > 1$ with $A=B$ or $A=B^\perp$~.
    \end{itemize}
    In these cases, the intersection necessarily includes non-scalar operators. 
\end{proposition}
Via Proposition \ref{Tphi}, the following is equivalent:
\begin{corollary}
\label{intersection}
    Let A and B be $a$ and $b$-dimensional subspaces of an $n$-dimensional Hilbert space $\mathcal{H}$, such that $n = k a = m b$ for integers $k,m$. Let $F_A$ be the set of von Neumann factor subalgebras $\mathcal{M}_A \equiv B(\mathcal{H})$ with $\mathcal{M}_A \equiv M_a(\mathbb{C})$ and $\mathcal{M}_A A = A$. Let $F_B$ be the set of von Neumann factor subalgebras $\mathcal{M}_B \subset B(\mathcal{H})$ with $\mathcal{M}_B \equiv M_b(\mathbb{C})$ and $\mathcal{M}_B B = B$. Then away from a measure zero set in $F_A \times F_B$,  $\mathcal{M}_A$ and $\mathcal{M}_B$ intersect only on multiples of the identity unless one of following holds:
    \begin{itemize}
    \item $A = \mathcal{H}$ and $ \dim B  > 1$~.
    \item $B = \mathcal{H}$ and $ \dim A  > 1$~.
    \item $\dim A = \dim B = \dim \mathcal{H}/2 > 1$ with $A=B$ or $A=B^\perp$~.
    \end{itemize}
    In these cases, the intersection necessarily includes non-scalar operators.
\end{corollary}
The proof follows the same strategy as for Proposition \ref{Union1}:
\begin{proof}
For $a=1$,  $\mathcal{M}_A$ and thus also $\mathcal{M}_A \cap \mathcal{M}_B$ includes only multiples of the identity. The same is true when $b=1$. For $a > 1$ and $b=n$, $\mathcal{M}_B$ is all of $B(\mathcal{H})$ so the intersection is $\mathcal{M}_A$ which includes more than multiples of the identity. The same conclusion holds for $b > 1$ and $a=n$. In the remainder, we assume $1 < a,b < n$ so that $k,m>1$.

Following the proof of the previous proposition, we have that
\[
\mathcal M_A=V(B(A)\otimes I_k)V^\dagger~,
\]
for $V:A\otimes \mathbb C^k\to \mathcal H$ and $
V(|\alpha\rangle \otimes e_1)=|\alpha\rangle$
and 
\[
\mathcal M_B=W(B(B)\otimes I_m)W^\dagger~,
\]
for $W:B\otimes \mathbb C^m\to \mathcal H$ and $
W(|\beta\rangle\otimes e_1)=|\beta\rangle$. As in the preceding proposition, $\mathcal{M}_A \cap \mathcal{M}_B$ is non-scalar if and only if it contains some non-trivial projection $P_R$ of rank $0 < R < n$. Since $P_R \in \mathcal{M}_A$, we can write it as $V (P_A \otimes I_k) V^\dagger$ for some projection $P_A$, and since $P_R \in \mathcal{M}_B$, we can write it as $W (P_B \otimes I_m) W^\dagger$ for some projection $P_B$. This shows that the rank $R$ is divisible by both $k$ and $m$. 

Let $C_V$ be the space of operators of the form $V(M_a(\mathbb{C}) \otimes I_k)V^\dagger$ for $V$ satisfying the conditions above and $C_W$ be the space of operators of the form $W (M_b(\mathbb{C}) \otimes I_m) W^\dagger$ for $W$ as above. The set $S_V$ of allowed $V$s is isomorphic to $U(n-a)$ since $V$ is a unitary on an $n$-dimensional space with a fixed action on an $a$-dimensional subspace. We thus have real dimension $\dim S_V = (n-a)^2$. Similarly, the set $S_W$ of allowed $W$s is isomorphic to $U(n-b)$ and has real dimension $\dim S_W = (n-b)^2$. In the space of operators on $\mathcal{H}$, let $\mathcal{P}_A$ be the set of all projectors of rank $R$ that appear in some $C_V$, $\mathcal{P}_B$ be the set of all projectors of rank $R$ that appear in some $C_W$, and $\Sigma = \mathcal{P}_A \cap \mathcal{P}_B$ be the intersection.

Consider the space of triples 
\[
Y = \{(V,W,P_R) \,|\, P_R \in C_V \cap C_W \} \; .
\]
This is a subset of $S_V \times S_W \times \Sigma$. We will show that (apart from the exceptional cases) the image $Y \to S_V \times S_W$ is a measure zero set and thus that generic $\mathcal{M}_A$ and $\mathcal{M}_B$ intersect only on multiples of the identity. We can think of $Y$ as having a base $\Sigma$ with fibers $F_A(P) = \{ V \in S_V \,|\, P \in C_V\}$ and $F_B(P) = \{ W \in S_W \, |\, P \in C_W\}$ in $S_V$ and $S_W$ respectively. We will see that the fiber dimensions are the same for any $P$, so this time we can simply write
\begin{equation}
\dim Y \le \dim \Sigma +  \dim F_A +  \dim F_B~.
\label{F1}
\end{equation}
We will show that the right hand side is smaller than $\dim (S_V \times S_W)$ apart from the exceptional cases.

We recall that the space $\mathcal{P}_A$ describes operators of the form
\[
P = V (P_A \otimes I_k) V^\dagger
\]
for some projector $P_A$ on $\mathbb{C}^a$ and some $V \in S_V$. We will show that this space 
is foliated by $S_V-$orbits labeled by an arbitrary rank $r=R/k$ projector in $M_a(\mathbb{C})$. Consider any specific projector $P = V_0 (P_A \otimes I_k) V_0^\dagger$. We will use $V_0 : A \otimes \mathbb{C}^k \to  \mathcal{H}$ to identify $A \otimes \mathbb{C}^k \cong \mathcal{H}$. In this case, using a basis $|\alpha_n\rangle \otimes e_i$ ordered as $|\alpha_1\rangle \otimes e_1, |\alpha_2\rangle \otimes e_1, \dots$, the matrix representation for $P$ takes a block form 
\[
P = \left(\begin{array}{ccc}P_A & 0 & 0 \cr 0 & \ddots & 0 \cr 0 & 0 & P_A \end{array}\right) \; ,
\]
and the allowed $V \in S_V$ take the form 
\[
V = \left(\begin{array}{cc}I_a & 0 \cr 0 & V_\perp \end{array}\right)~, \qquad V_\perp \in U(n-a) \; .
\]
The action $P \to V P V^\dagger$ does not affect the upper left block, but can bring the full matrix to the standard form
\[
P_0(P_A) = \left(\begin{array}{ccc}P_A & 0 & 0 \cr 0 & I_{(k-1)r} & 0 \cr 0 & 0 & 0_{(k-1)(a-r)} \end{array}\right) \; .
\]

We now compute the fiber dimension $\dim F_A(P)$ and show that it is independent of $P$. Let $P = V_0 (P_A \otimes I_k) V_0^\dagger$. The fiber $F_A(P)$ is the set of $V$s for which there is a $\hat{P}_A$ satisfying $V(\hat{P}_A \otimes I_k)V^\dagger = P$. But we must have $\hat{P}_A = P_A$, since $V$ does not change the upper-left block. In this case, we have $P = V_0 (P_A \otimes I_k) V_0^\dagger = V (P_A \otimes I_k) V^\dagger$ so the allowed $V$s are those that satisfy $(V V_0^\dagger) P (V V_0^\dagger)^\dagger = P$. Thus $F_A(P)$ is the subset of $S_V$ where $V V_0^\dagger$ is in the stabilizer of $P$. The dimension $\dim F_A(P)$ is therefore the dimension of the stabilizer of $P$ in $S_V$. This is the same as the dimension $D_A$ of the stabilizer of the standard form $P_0(P_A)$, which is
\[
\dim F_A = D_A = (k-1)^2 r^2 + (k-1)^2 (a-r)^2 \; .
\]
A similar computation gives
\[
\dim F_B = D_B = (m-1)^2 s^2 + (m-1)^2 (b-s)^2 \; ,
\]
where $s=R/m$.

For our dimension comparison, we also need to bound $\dim \Sigma$. We have
\begin{equation}
\dim \Sigma \le \min(\dim {\cal P}_A,\dim {\cal P}_B) = \dim {\cal P}_A + \dim {\cal P}_B - \max(\dim {\cal P}_A,\dim {\cal P}_B) \; .
\label{F3}
\end{equation}
Further, since for any $P \in {\cal P}_A$, we can write $P = V(P_0(P_A))V^\dagger$, we have
\begin{equation}
\dim {\cal P}_A = \dim S_V - D_A  + \dim S_A \; ,
\label{F4}
\end{equation}
where $S_A$ is the set of $P_A$s. Since $P_A$ is an arbitrary rank $r$ projector in $M_a(\mathbb{C})$, and these form a single $U(a)$ orbit, the dimension is
\[
\dim S_A = a^2 - \dim \text{Stab}(P_A) = a^2 - r^2 - (a-r)^2 = 2r (a-r) \; .
\]
Combining equations (\ref{F1}), (\ref{F3}), and (\ref{F4}), we find that
\[
\dim Y \le [\dim S_V + \dim S_W] + [\dim S_A + \dim S_B - \max(\dim {\cal P}_A,\dim {\cal P}_B)]~.
\]
Hence, the dimension of $Y$ will be smaller than the dimension of $S_V \times S_W$ provided that
\[
\max(\dim {\cal P}_A,\dim {\cal P}_B) > \dim S_A + \dim S_B \; .
\]
This will hold if
\begin{equation}
{1 \over 2}(\dim {\cal P}_A+\dim {\cal P}_B) > \dim S_A + \dim S_B \; .
\label{final2}
\end{equation}
Using the explicit results above, we have that
\[
{1 \over 2} \dim {\cal P}_A - \dim S_A = {1 \over 2} (\dim S_V - D_A - \dim S_A) = r(a-r)((k-1)^2 - 1) \; ,
\]
with similar results for the $B$ side. Since we have $k>1$ and $0 < r < a$, the right side here is non-negative, vanishing only for $k=2$, i.e. for $\dim A = \dim \mathcal{H}/2$. A parallel argument shows that ${1 \over 2} \dim {\cal P}_B - \dim S_B$ is positive except in the case $m=2$ i.e. for $\dim B = \dim \mathcal{H}/2$.
Thus, (\ref{final2}) can fail to be true only if 
$\dim A = \dim B = \dim \mathcal{H}/2$. In all other cases, the image of $Y$ in $S_V \times S_W$ is measure zero so typical choices for $\mathcal{M}_A$ and $\mathcal{M}_B$ will intersect only on multiples of the identity.

We now analyze the special case separately. In this special case, we have $a = b = n/2$ and so $k=m=2$. The constraints $R=kr=ms$ on the rank of the projector give $r=s=R/2$. Using the results above, we have
\[
\dim S_V = \dim S_W = a^2~, \qquad \dim S_A = \dim S_B = 2r(a-r)~, \qquad D_A = D_B = r^2 + (a-r)^2 \; .
\]
From (\ref{F4}), we have 
\[
\dim \mathcal{P}_A = \dim \mathcal{P}_B = 4r(a-r) \; .
\]
From (\ref{F1}) we can say that the dimension of $Y$ will be less than that of $S_V \times S_W$ if
\[
\dim \Sigma + D_A + D_B < \dim S_V + \dim S_W \quad \Longleftrightarrow \quad \dim \Sigma < 4r(a-r) \; .
\]
Since $\Sigma = \mathcal{P}_A \cap \mathcal{P}_B$ which both have dimension $4r(a-r)$, this inequality will hold unless $\mathcal{P}_A$ and $\mathcal{P}_B$ have the same dimension as their intersection. We will show that this happens if and only if $B = A$ or $B= A^\perp$.

From our discussion above about the explicit form of $P$, we have in this case that in a basis where the first $a$ elements span the subspace $A$, the operators in $\mathcal{P}_A$ take the form $\text{diag}(P_1,P_2)$,
where the projectors $P_1$ and $P_2$ each have rank $0<r<a$. Thus, the space  $\mathcal P_A$ is equivalent to a product $\operatorname{Gr}(r,A)\times \operatorname{Gr}(r,A^\perp)$ of Grassmanians, where $\operatorname{Gr}(r,A)$ is the space of all $r$-dimensional linear subspaces of $A$ (in this case, the spaces $P_1 A$ and $P_2 A$). As a result, we can say that $\mathcal P_A$ is an irreducible real algebraic manifold (i.e. it can be specified by a set of polynomial equations and cannot be divided into multiple non-trivial parts that are each specified by polynomial equations). Since the intersection $\mathcal P_A \cap \mathcal P_B$ is the same dimension as $\mathcal P_A$, it includes some open set of $\mathcal P_A$. Every element of $\mathcal P_B$ commutes with $\Pi_B$, so the polynomial equation $[P,\Pi_B]=0$ on the elements of $P$ holds on $\mathcal P_A \cap \mathcal P_B$. But polynomial equations that hold on an open subset of an irreducible algebraic manifold must hold everywhere on that manifold, so $[P,\Pi_B]=0$ for all $P \in \mathcal{P}_A$.

Thus, for all rank $r$ projectors $P_1$ on $A$ and $P_2$ on $A^\perp$, we have
\[
[\Pi_B, P_1 \oplus P_2]=0~.
\]
Varying $P_1$ while holding $P_2$ fixed shows that $\Pi_B$ commutes
with all differences of rank $r$ projectors on $A$. Such differences span the full space of traceless Hermitian operators on $A$. Similarly, $\Pi_B$ commutes with all
traceless self-adjoint operators on $A^\perp$. This implies that in our basis adapted to $\mathcal H=A\oplus A^\perp$,
\[
\Pi_B=
\begin{pmatrix}
\lambda I_A & 0\\
0 & \mu I_{A^\perp}
\end{pmatrix}.
\]
for some constants $\lambda$ and $\mu$.
Since $\Pi_B$ is a projection and $\dim B=\dim A$, either
$(\lambda,\mu)=(1,0)$ or $(\lambda, \mu) = (0,1)$. Hence $B=A$ or $B=A^\perp$, as we wished to show. In all other cases, we must have $\dim Y < \dim S_V \times S_W$ so $\mathcal{M}_A \cap \mathcal{M}_B = \mathbb{C} I$ except for a measure zero set in $S_V \times S_W$.

Finally, we analyze the exceptional cases $A=B$ and $A=B^\perp$. In either of these cases, using the same basis as above, any algebra $\mathcal{M}_A$ will correspond to the set of operators of the form $\text{diag}(\mathcal{O}, U_1  {\cal O} U_1^\dagger)$ for some unitary $U_1 \in U(a)$, while any algebra $\mathcal{M}_B$ will correspond to the set of operators of the form $\text{diag}(\mathcal{O}, U_2  {\cal O} U_2^\dagger)$
for some unitary $U_2 \in U(a)$. Taking ${\cal O}$ to be a projector to one of the eigenvectors of $U_2^\dagger U_1$, we have that $U_2^\dagger U_1 {\cal O} = {\cal O} U_2^\dagger U_1$, so $\text{diag}(\mathcal{O}, U_1  {\cal O} U_1^\dagger)=\text{diag}(\mathcal{O}, U_2  {\cal O} U_2^\dagger)$
is a projector of rank 2. Thus, we always have a non-scalar operator common to $\mathcal{M}_A$ and $\mathcal{M}_B$.
\end{proof}

\subsubsection{Commutativity of $\mathcal{M}_A$ and $\mathcal{M}_B$}

The next proposition characterizes when the algebras $\mathcal{M}_A$ and $\mathcal{M}_B$ can be chosen to commute with one another.
\begin{proposition}
\label{commuting1}
Let $\mathcal H$ be an $n$-dimensional Hilbert space, and let
$A,B\subset \mathcal H$ have dimensions $a$ and $b$. Let $P_A$ and $P_B$ be the orthogonal
projections onto $A$ and $B$. There exist commuting von Neumann factor subalgebras
\[
  \mathcal{M}_A \equiv M_a(\mathbb{C})~, \quad \mathcal{M}_A A = A~, \qquad \mathcal{M}_B \equiv M_b(\mathbb{C})~, \quad \mathcal{M}_B B = B~,
\]
if and only if one of the following holds 
\begin{itemize}
    \item 
     $n = abr, r=1$ and $\text{rank}  (P_AP_B)=1$ and $\|P_AP_B\|_1=1$~, \item 
     $n = abr, r=2,3,\dots$ and $\text{rank}  (P_AP_B) \le r$ and $\|P_AP_B\|_1 \le 1$~,
\end{itemize}
where $\|\cdot\|_1$ denotes the trace norm.
\end{proposition}
The necessary and sufficient conditions in the proposition can be equivalently formulated in terms of the principal angles between $A$ and $B$, since $\text{rank}(P_A P_B)$ is equal to the number of principal angles with nonzero $\cos \theta$ and 
\[
||P_A P_B||_1 = \sum_i \cos \theta_i \; .
\]
\begin{proof}
Suppose that algebras with the stated properties exist. Then $k = n/a$ and $m = n/b$ must be integers since the algebras are factors. We can define an orthonormal basis $a_n \otimes e_i$ of $\mathcal{H}$ with $a_n \otimes e_1$ a basis for $A$ such that $\mathcal{M}_A = M_a(\mathbb{C}) \otimes I_k$ and $\mathcal{M}_A' = I_a \otimes M_k(\mathbb{C})$. Since $\mathcal{M}_B \subset \mathcal{M}_A'$, the operators in $\mathcal{M}_B$ must be of the form $I_a \otimes \hat{\mathcal{B}}$ for some subalgebra $\hat{\mathcal{B}}$ of $M_k(\mathbb{C})$. The algebra $\hat{\mathcal{B}}$ must be a subfactor of $M_k(\mathbb{C})$, since any non-scalar operator ${\cal O}$ in its center would give a non-scalar operator $I_a \otimes {\cal O}$ in the center of $\mathcal{M}_B$. Thus, $k = b r$ for some positive integer $r$ and so $n = a b r$. We can choose a basis $c_i \otimes b_m$ for $\mathbb{C}^k$ so that $a_n \otimes c_i \otimes b_m$ is a basis for $\cal{H}$ in which $\mathcal{M}_A = M_a(\mathbb{C}) \otimes I_r \otimes I_b$ and $\mathcal{M}_B = I_a \otimes I_r \otimes M_b(\mathbb{C})$. We can write
\[
e_1 = \sum_{i,m}U_{i m} c_i  \otimes b_m~,  \; 
\]
for some $U_{im}$ with $\sum_{i,m} |U_{im}|^2 = 1$. The projector $P_B$ has rank $b$ and commutes with $\mathcal{M}_B$, so we can write it as
\[
P_B = \hat{P}_B \otimes I_b~,
\]
for some projector $\hat{P}_B$ of rank 1. Letting $f$ be a normalized eigenvector of $\hat{P}_B$ with eigenvalue 1, we can write
\[
f = \sum_{n,i}V_{n i} a_n \otimes c_i~, 
\]
for some $V_{ni}$ with $\sum_{i,n} |V_{ni}|^2 = 1$. 

We now have an orthonormal bases
\[
\alpha_n \equiv a_n \otimes e_1 = \sum_{i,m} U_{i m} a_n \otimes  c_i  \otimes b_m~, \qquad \beta_m \equiv f\otimes b_m = \sum_{n,i} V_{n i} a_n \otimes c_i \otimes b_m~,
\]
for $A$ and $B$ respectively. Let $C$ be the $a\times b$ overlap matrix between the bases:
\[
  C_{nm} \equiv \langle \alpha_n|\beta_m\rangle = (V U^{\dagger})_{nm}~.
\]
Viewing $U^{\dagger}$ as a linear map from $\mathbb{C}^m $ to $\mathbb{C}^r$, its rank is at most $r$, so we have:
\[
  \operatorname{rank}(C)\le r~.
\]
To compute the trace norm $\|C\|_1 =\tr \sqrt{C^{\dagger}C}$, we define $u = \sqrt{U U^\dagger}$ and $v = \sqrt{V^\dagger V}$, and by the Cauchy-Schwarz inequality it follows that:
\[
  \|C\|_1 = \sum_{ij} u_{ij} v_{ij} \le \sqrt{ \sum_{ij} |u_{ij}|^2} \sqrt{ \sum_{ij} |v_{ij}|^2}
  =
  1~.
\]
Since 
\[
P_A P_B = \sum_{n,m} |\alpha_n \rangle \langle \alpha_n | \beta_m \rangle \langle \beta_m| \; ,
\]
the singular values of $C$ (i.e. the eigenvalues of $\sqrt{C^{\dagger}C}$) are precisely the singular values of
$P_AP_B$. Hence, we have shown that
\[
  \operatorname{rank}(P_AP_B)\le r
  \qquad\text{and}\qquad
  \|P_AP_B\|_1\le 1~.
\]
If $r=1$, then $U$ is a single row and $V$ is a single column, each of norm
one. Hence $VU^*$ has rank one and trace norm one. Thus, in the case $r=1$ we get the stronger condition
\[
  \operatorname{rank}(P_AP_B)=1\qquad\text{and}\qquad
  \|P_AP_B\|_1=1~.
\]
This proves necessity.

We now prove sufficiency. Suppose that one of the two conditions holds (where $r=1$ corresponds to the first condition, and $r \ge 2$ corresponds to the second). Let $q = \operatorname{rank}(P_A P_B)$ and let $\sigma_i,\, i=1, \ldots, q$ be the non-zero singular values of $P_A P_B$, so that $\sum_i \sigma_i = ||P_A P_B||_1$. These are precisely the cosines of the principal angles between $A$ and $B$, so we can choose a basis $\alpha_i$ for $A$ and $\beta_i$ for $B$ so that
\[
\langle \alpha_i | \beta_i \rangle = \sigma_i~, \qquad i = 1,\dots , q~,
\]
and all other inner products vanish. Now, suppose that there exist vectors $w, u \in \mathbb{R}^r$ with non-negative entries such that 
\begin{equation}
\label{properties}
w_i u_i = \sigma_i~, \qquad \sum_{i=1}^r w_i^2 = \sum_{i=1}^r u_i^2 = 1~,
\end{equation}
where we define $\sigma_i = 0$ for $i > q$.  We introduce a reference Hilbert space $\hat{\mathcal{H}} = \mathbb{C}^a \otimes \mathbb{C}^r \otimes \mathbb{C}^b$ endowed with the orthonormal basis $\hat{a}_n \otimes \hat{c}_i \otimes \hat{b}_m$ and define
\[
\hat{\alpha}_n \equiv  \hat{a}_n \otimes \left( \sum_{i} u_i \hat{c}_i \otimes \hat{b}_i \right) \equiv \hat{a}_n \otimes \hat{e}_1~,\qquad 
\hat{\beta}_m \equiv  \left( \sum_{i} w_{i} \hat{a}_i \otimes \hat{c}_i \right) \otimes \hat{b}_m \equiv \hat{f} \otimes \hat{b}_m~.
\]
The properties (\ref{properties}) ensure that the inner products between the $\hat{\alpha}$ and $\hat{\beta}$ vectors match those of the $\alpha$ and $\beta$ vectors. Therefore, there exists a unitary $W$ mapping $\hat{\mathcal{H}}$ to $\mathcal{H}$ with the property that $W \hat{\alpha}_n = \alpha_n$ and $W \hat{\beta}_n = \beta_n$. It is straightforward to check that the algebras $\mathcal{M}_A = W(\mathcal{M}_a(\mathbb{C}) \otimes I_r \otimes I_b) W^\dagger$ and $\mathcal{M}_B = W(I_a \otimes I_r \otimes \mathcal{M}_B(\mathbb{C})) W^\dagger$ then have the desired properties.

It remains to show that there are vectors $w, u \in \mathbb{R}^r$ such that (\ref{properties}) holds:
\begin{itemize}
    \item If $r=1$, we have $q=1$ and $\sigma_1 = 1$, so we can take $w_1 = u_1 = 1$.
    \item For $r \ge 2$, $q=0,1$, we can take $u_1 = \sigma_1, u_2 = \sqrt{1- \sigma_1^2}, u_{n >2} = 0$ and $w_1 = 1, w_{n > 1}=0$.
    \item 
    For $r \ge 2, q \ge 2$ we take $w_i = \sigma_i/u_i$ for $u_i \ne 0$ and $w_i = 0$ for $u_i = 0$. We need to show that there are simultaneous solutions to 
\[
\sum_{i=1}^q u_i^2 = 1~, \qquad \sum_{i=1}^q {\sigma_i^2 \over u_i^2} = 1~,
\]
with $u_i > 0$. The first equation defines the boundary of the unit sphere. We will show that the surface defined by the second equation has a continuous branch with points both inside/on and outside this sphere, and therefore must intersect the sphere. Let $S = \sum_{i=1}^q \sigma_i$ and consider $u_i = u_i^0 \equiv \sqrt{S \sigma_i}$. Then we have
\[
\sum_{i=1}^q {\sigma_i^2 \over u_i^2} = \sum_{i=1}^q {\sigma_i^2 \over S \sigma_i} = 1~, \qquad \sum_{i=1}^q u_i^2 = S \sum_{i=1}^q \sigma_i = S^2 \le 1~,
\]
so this set of $u_i$s corresponds to a point inside/on the unit sphere satisfying the second equation. Now, we note that it is possible to vary $u_1$ and $u_2$ continuously so that the second equation remains satisfied and $\vec{u}=(u_1,\ldots, u_q)$ moves outside the unit sphere. Indeed, for positive $u_1,u_2$ the curve
\[
{\sigma_1^2 \over u_1^2} + {\sigma_2^2 \over u_2^2} = {\sigma_1^2 \over (u_1^0)^2} + {\sigma_2^2 \over (u_2^0)^2}~,
\]
is continuous and has points with arbitrarily large $u_1$.
\end{itemize}
\end{proof}

\section{Factor codes for diamond networks}

In this section, we consider quantum code networks whose corresponding Hasse diagram has a diamond shape as in Figure \ref{fig:Hasse} (right). In this case, we will see that even if the required divisibility conditions on Hilbert space dimensions hold, it is not always possible to promote the code network to a factor code network.

As a very simple explicit counterexample, consider Hilbert spaces of dimension 2,4,4, and 8 with isometries defined as shown in Figure \ref{fig:CE2448}.

\begin{figure}
    \centering
    \includegraphics[width=0.4\linewidth]{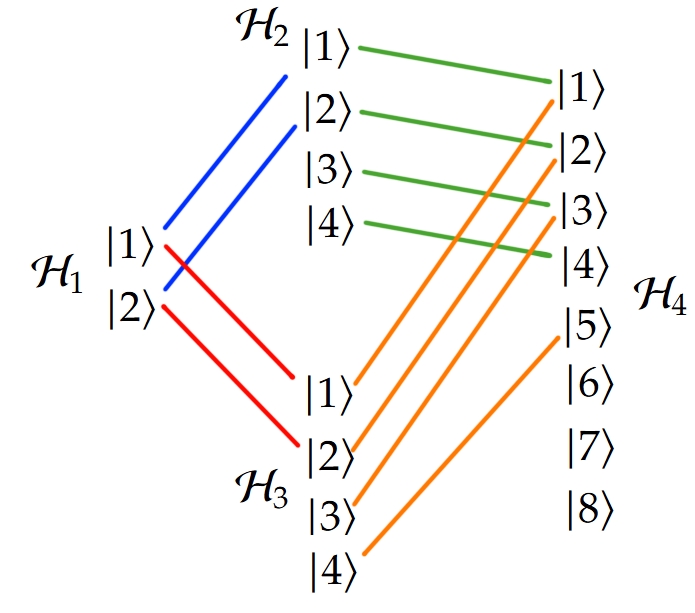}
    \caption{A set of isometries for Hilbert spaces of dimensions 2,4,4, and 8. These isometries cannot be upgraded to a set of operator maps satisfying the properties in Proposition 1.}
    \label{fig:CE2448}
\end{figure}

Regardless of the choices for our auxiliary isometries, an operator in $B(\mathcal{H}_1)$ maps via $\varphi_{42} \circ \varphi_{21}$ to an operator that leaves invariant the subspace spanned by $|3 \rangle$ and $|4 \rangle$. On the other hand, we can see that $\varphi_{43} \circ \varphi_{31} ({\cal O})$ leaves invariant the subspace spanned by $|3 \rangle$ and $|5 \rangle$, but for an appropriate choice of ${\cal O}$ we can realize any operator acting on this subspace. In particular, there will be operators that take $|3 \rangle$ into linear combinations with $|5 \rangle$ and thus out of the subspace spanned by $|3 \rangle$ and $|4 \rangle$. This shows that $\varphi_{43} \circ \varphi_{31}(B(\mathcal{H}_1))$ cannot be the same as $\varphi_{42} \circ \varphi_{21}(B(\mathcal{H}_1))$, since we have already seen that the latter leaves the 3-4 subspace invariant. Thus, the construction will not be possible for a diamond-shaped Hasse diagram with certain isometries. In general, the success of the construction depends on how the subspaces $T_{42} \mathcal{H}_2$ and $T_{43} \mathcal{H}_3$ sit inside $\mathcal{H}_4$ (see Figure \ref{fig:DiamondMaps}).

\begin{figure}
    \centering
\includegraphics[width=0.5\linewidth]{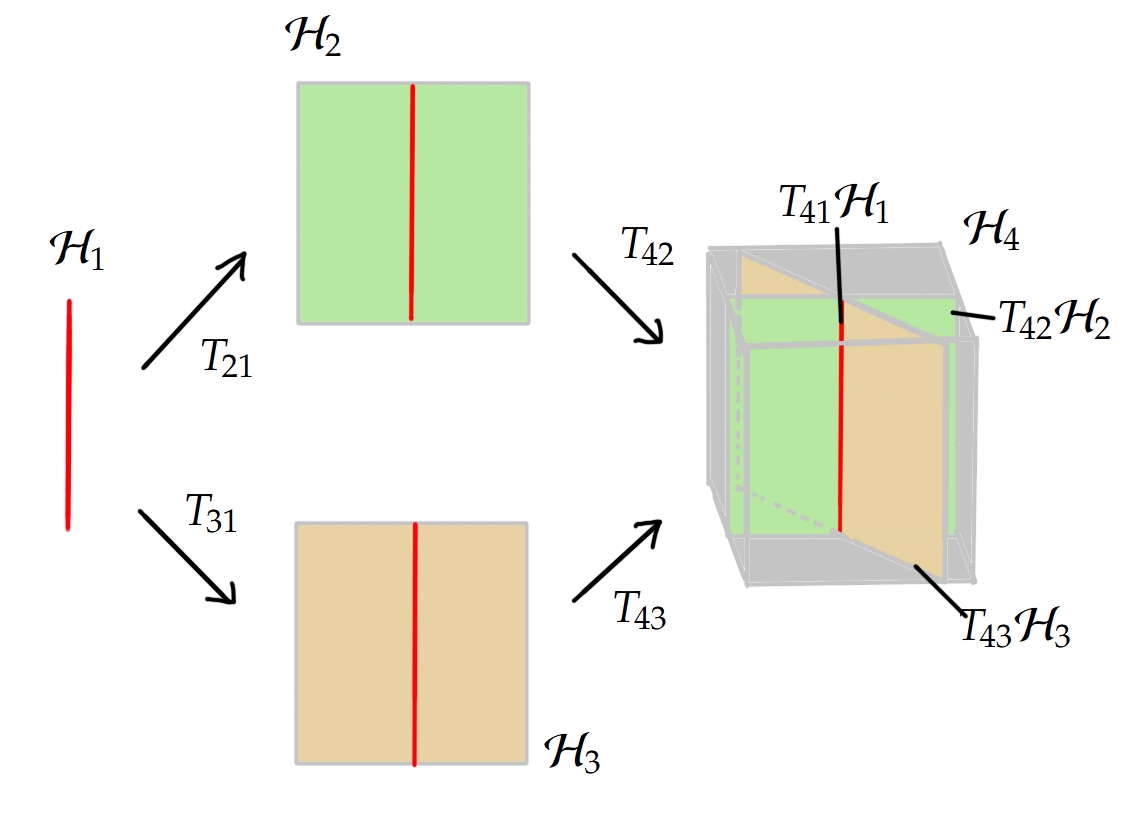}
    \caption{Hilbert spaces and isometries for a diamond-shaped Hasse diagram. The isometries can be promoted to a set of operator maps satisfying the properties of Proposition 1 if and only if the principal angles between the two subspaces  $T_{42} \mathcal{H}_2$ and $T_{43} \mathcal{H}_3$ have degeneracies that are multiples of $\dim(\mathcal{H}_1)$.}
    \label{fig:DiamondMaps}
\end{figure}

We will now show that the construction of a consistent set of  operator maps is possible if and only if the principal angles between $T_{42} \mathcal{H}_2$ and $T_{43} \mathcal{H}_3$ each come with multiplicity that is a multiple of $\dim(\mathcal{H}_1)$. 

\begin{proposition}{\it Consider a code network associated with a set $S = {1,2,3,4}$ with partial ordering $1 < 2,3 < 4$. It is possible to define unital operator maps promoting this to a factor code network if and only if the principal angles between $T_{42} \mathcal{H}_2$ and $T_{43} \mathcal{H}_3$ have multiplicities which are multiples of $\dim(\mathcal{H}_1)$.}
\label{diamond}
\end{proposition}
\begin{proof}
    We first show that the condition is necessary. Suppose that the various auxiliary isometries have been chosen such that
    \[
    \varphi_{ji}(\mathcal{O}) = \sum_a T^{(a)}_{ji}{\cal O} (T^{(a)}_{ji})^\dagger \; .
    \]
    Then using $(T^{(a)}_{ij})^\dagger T^{(b)}_{ij} = \delta_{ab}$, it follows that $P_2 = T_{42} T_{42}^\dagger$ commutes with $\varphi_{42} \circ \varphi_{21}({\cal O}) =\varphi_{41}({\cal O})$ for any ${\cal O} \in B(\mathcal{H}_1)$. By the same argument, $P_3 = T_{43} T_{43}^\dagger$ also commutes with $\varphi_{41}({\cal O})$ for any ${\cal O}$. Since the algebra $\mathcal{M}_1^{(4)} = \varphi_{41}(B(\mathcal{H}_1))$ is a factor, we must have a decomposition $\mathcal{H}_4 \equiv \mathcal{H}_1 \otimes I$ such that $\mathcal{M}_1^{(4)}$ corresponds to operators of the form ${\cal O} \otimes I$. The operators $P_2$ and $P_3$ are in the commutant, so are of the form $P_i = I \otimes \hat{P}_i$. It follows that all the eigenvalues of $P_2 P_3 P_2 = I \otimes \hat{P}_2 \hat{P}_3 \hat{P}_2$ (and thus the principal angles between $T_{42} \mathcal{H}_2$ and $T_{43} \mathcal{H}_3$) have multiplicities that are multiples of $\dim({\cal H}_1)$.

    Next, we show sufficiency. If all principal angles have degeneracies that are multiples of $\dim({\cal H}_1)$, we have an identification 
    \[
    \mathcal{H}_4 \equiv \mathcal{H}_1 \otimes \mathcal{K}~,
    \]
    where $T_{41} \mathcal{H}_1 \equiv \mathcal{H}_1 \otimes e_1$, $T_{42} \mathcal{H}_2 \equiv \mathcal{H}_1 \otimes \mathcal{K}_2$ and $T_{43} \mathcal{H}_3 \equiv \mathcal{H}_1 \otimes \mathcal{K}_3$ for some subspaces $\mathcal{K}_2,\mathcal{K}_3 \subset \mathcal{K}$ and $e_1 \in \mathcal{K}_2 \cap \mathcal{K}_3$. 
For the path $1 \rightarrow 2 \rightarrow 4$, we can define auxiliary isometries $T_{21}^{(a)}$ by noting $\mathcal{H}_2 = \mathcal{H}_1 \otimes \mathcal{K}_2$ and 
\[
T_{21}^{(a)} |v\rangle = |v\rangle \otimes|a \rangle \; 
\]
for some orthonormal basis $\{|a \rangle\}$ of $\mathcal{K}_2$. Next, since $\dim(\mathcal{K}_2)$ divides $\dim(\mathcal{K})$, we can identify $\mathcal{K} = \mathcal{K}_2 \otimes \bar{\mathcal{K}}_2$. Choosing an orthonormal basis $|i \rangle$ for $\bar{\mathcal{K}}_2$, we can take
\[
T_{42}^{(i)} |w\rangle =|w \rangle \otimes |i \rangle \; .
\]
Then the full operator map defined along the $1 \to 2 \to 4$ path gives
\[
\varphi_{42} \circ \varphi_{21} ({\cal O}) = \sum_{a,i} T_{42}^{(i)} T_{21}^{(a)} {\cal O} (T_{21}^{(a)})^\dagger (T_{42}^{(i)})^\dagger =\sum_{a,i}  {\cal O} \otimes (|a \rangle \otimes |i \rangle) (\langle i | \otimes \langle a |) = {\cal O} \otimes I_\mathcal{K} \; .
\]
We can define isometries along the $1 \to 3 \to 4$ path in a similar way, giving
\[
\varphi_{43} \circ \varphi_{31} ({\cal O})  ={\cal O} \otimes I_\mathcal{K} = \varphi_{42} \circ \varphi_{21} ({\cal O}) \; ,
\]
so the compatibility condition is satisfied.
\end{proof}

\subsection{When does $\mathcal{M}_2 \cap \mathcal{M}_3$ equal $\mathcal{M}_1$?}

We have seen in Proposition \ref{prop:intersection} that when we have two code subspaces in the same Hilbert space, generic choices for the extension to a factor code will typically have the pair of associated subalgebras intersect only on multiples of the identity except in some exceptional cases.

However, in cases covered by Proposition \ref{diamond} above, any factor code network associated with a diamond-shaped graph gives a set of subalgebras $\mathcal{M}^{(4)}_1 \subset \mathcal{M}^{(4)}_2,\mathcal{M}^{(4)}_3 \subset \mathcal{M}^{(4)}_4$. Here, we will see that the typical
intersection between $\mathcal{M}^{(4)}_2$ and $\mathcal{M}^{(4)}_3$ is not only multiples of the identity but instead the common subalgebra $\mathcal{M}^{(4)}_1$.

We have
\begin{proposition}
\label{intersection1}
    Suppose we have $\mathcal{H}_i$ and isometries $T_{ji}$ satisfying the conditions of Proposition \ref{diamond}. Then by an appropriate choice of isometries, we can ensure that $\mathcal{M}^{(4)}_1 = \mathcal{M}^{(4)}_2 \cap \mathcal{M}^{(4)}_3$ unless one of following holds:
    \begin{itemize}
    \item $\dim \mathcal{H}_1 < \dim \mathcal{H}_2 \le \dim  \mathcal{H}_3 = \dim \mathcal{H}_4$~.
    \item $\dim \mathcal{H}_1 < \dim\mathcal{H}_3 \le \dim \mathcal{H}_2 = \dim \mathcal{H}_4$~.
    \item $\dim \mathcal{H}_1 < \dim \mathcal{H}_2 = \dim \mathcal{H}_3 = \dim \mathcal{H}_4/2$ with $T_{42} \mathcal{H}_2=T_{43} \mathcal{H}_3$ or $T_{42} \mathcal{H}_2=T_{43} \mathcal{H}_3^\perp$~.
    \end{itemize}
    Furthermore, in all of these cases, generic choices of auxiliary isometries lead to the desired intersection property.
\end{proposition}

\begin{proof}
    As we saw in the proof of Proposition \ref{diamond}, we can identify $\mathcal{H}_4 = \mathcal{H}_1 \otimes \mathcal{K}$ with $T_{42} \mathcal{H}_2 = \mathcal{H}_1 \otimes A$, $T_{43} \mathcal{H}_3 = \mathcal{H}_1 \otimes B$ and $T_{41} \mathcal{H}_1 = \mathcal{H}_1 \otimes \mathbb{C} e_1$ where $A$ and $B$ are subspaces of $\mathcal{K}$ whose dimension divides $\dim \mathcal{K}$ and $e_1 \in A \cap B$. With this decomposition, $\mathcal{M}^{(4)}_1$ is the algebra $B(\mathcal{H}_1) \otimes \mathcal{K}$, and $\mathcal{M}^{(4)}_2$ and $\mathcal{M}^{(4)}_3$ are algebras of operators $B(\mathcal{H}_1) \otimes \hat{\mathcal{M}}_{2,3}$ where 
    \[
    \hat{\mathcal{M}}_{2} = \{\oplus_{i=1}^k u_i {\cal O} u_i^\dagger \,|\, {\cal O} \in M_a(\mathbb{C}), u_i \in U(\dim A)\}
    \]
    for some decomposition
    $\mathcal{K} = A_1 \oplus \cdots \oplus A_k$ with $A_1 = A$ and 
     \[
    \hat{\mathcal{M}}_{3} = \{\oplus_{i=1}^k u_i {\cal O} u_i^\dagger \,| \, {\cal O} \in M_b(\mathbb{C}), u_i \in U(\dim B)\}
    \]
    for some decomposition with $B_1 = B$. The algebra $\hat{\mathcal{M}}_{2}$ is isomorphic to $M_{a}(\mathbb{C})$ and satisfies $\hat{\mathcal{M}}_{2} A = A$ while the algebra $\hat{\mathcal{M}}_{3}$ is isomorphic to $M_{b}(\mathbb{C})$ and satisfies $\hat{\mathcal{M}}_{2} B = B$. 
   The property $\mathcal{M}^{(4)}_1 = \mathcal{M}^{(4)}_2 \cap \mathcal{M}^{(4)}_3$ holds if and only if $\hat{\mathcal{M}}_{2} \cap \hat{\mathcal{M}}_{3} = \mathbb{C} I$. Thus, we can apply Proposition \ref{intersection} for ${\cal H} = \mathcal{K}$ to immediately obtain the stated result.
\end{proof}

\subsection{Conditions for a nondegenerate commuting square}

When the promotion of a diamond-shaped code network to a factor code network is possible, we have algebra inclusions ${\cal M}_1^{(4)} \subset {\cal M}_2^{(4)},{\cal M}_3^{(4)} \subset {\cal M}_4^{(4)}$, and we can define conditional expectations
\[
E_{j \to i} : {\cal M}_j^{(4)} \to {\cal M}_1^{(4)}
\]
which in this finite-dimensional setting can be understood  geometrically as a projection to ${\cal M}_1^{(4)}$ considered as a subspace of ${\cal M}_j^{(4)}$. 

In this situation with a diamond of nested algebras, it is in various contexts useful to understand when these satisfy the {\it commuting square} condition \cite{Popa1983,JonesSunder1997}, that
\[
E_{4 \to 3} \circ E_{4 \to 2} = E_{4 \to 2} \circ E_{4 \to 3} = E_{4 \to 1} \; ,
\]
and the associated { \it co-commuting square condition } that the commutants of the algebras within ${\cal M}_4^{(4)}$ form a commuting square. For example, the commuting square condition appears in the algebraic version of the strong-subadditivity inequality for quantum entropies \cite{Petz1991,GaoJungeLaRacuente2019}.

Thinking in terms of the conditional expectations as projectors to subspaces, we can see that the commuting square condition is equivalent to the requirement that
\[
{\cal M}_2^{(4)} \ominus {\cal M}_1^{(4)} \perp {\cal M}_3^{(4)} \ominus {\cal M}_1^{(4)} \; .
\]
That is, we require that the orthogonal complement of ${\cal M}_1^{(4)}$ in ${\cal M}_2^{(4)}$ is perpendicular to the orthogonal complement of ${\cal M}_1^{(4)}$ in ${\cal M}_3^{(4)}$. This implies in particular that ${\cal M}_1^{(4)}$ is the intersection of ${\cal M}_2^{(4)}$ and ${\cal M}_3^{(4)}$.

In this finite-dimensional setting, we can characterize precisely the structure of a commuting and co-commuting square of von Neumann factors. We prove the following in Appendix \ref{sec:NDcomm}:
\begin{proposition}
\label{NDnormal}
    If $A \subset B,C \subset D$ is a nondegenerate commuting square of von Neumann factors in a finite dimensional von Neumann algebra, then up to unitary equivalence, we have
    \[
A \equiv M_{n}(\mathbb{C}) \otimes I_k \otimes I_m\;, \; B \equiv M_{n}(\mathbb{C}) \otimes M_{k}(\mathbb{C}) \otimes I_m \;, \; D \equiv M_{n}(\mathbb{C}) \otimes M_{k}(\mathbb{C}) \otimes M_{m}(\mathbb{C}) \;, \;  C \equiv M_{n}(\mathbb{C}) \otimes U (I_k \otimes M_{m}(\mathbb{C})) U^\dagger \;, \; 
\]
where $U$ is a bi-unitary matrix, i.e. a unitary operator on $\mathbb{C}^k \otimes \mathbb{C}^m$ represented by a unitary matrix $U_{i \alpha, j \beta}$ such that $U_{i \beta, j \alpha}$ is also unitary.
\end{proposition}

When the conditions of Proposition \ref{diamond} are satisfied, we can obtain further necessary and sufficient conditions on the isometries such that the diamond of algebras form a commuting and co-commuting square (also called a nondegenerate commuting square). Specifically, we require that the principal angles defined there are 0 with multiplicity $\dim \mathcal{H}_1$ and $\pi/2$ for the remaining ones.

\begin{proposition}
\label{prop:diamond_commuting_square}
Consider Hilbert spaces
\[
\mathcal H_1,\mathcal H_2,\mathcal H_3,\mathcal H_4
\]
ordered as \(1<2,3<4\).  Let
\[
T_{ji}:\mathcal H_i\to \mathcal H_j
\]
be isometries satisfying
\[
T_{41}=T_{42}T_{21}=T_{43}T_{31}~.
\]
Let $P_i$ denote the orthogonal projector onto \(T_{4i}\mathcal H_i\subset \mathcal H_4\), for \(i=1,2,3\).  Set
\[
n=\dim\mathcal H_1~,\qquad
k=\frac{\dim\mathcal H_2}{\dim\mathcal H_1}~,\qquad
m=\frac{\dim\mathcal H_3}{\dim\mathcal H_1}~.
\]
Then the isometries can be completed by auxiliary isometries so that the associated factor algebras
\[
\mathcal M_1^{(4)}\subset \mathcal M_2^{(4)},\mathcal M_3^{(4)}
\subset \mathcal M_4^{(4)}=B(\mathcal H_4)
\]
form a nondegenerate commuting square if and only if
\[
\dim\mathcal H_4
=
\frac{\dim\mathcal H_2\,\dim\mathcal H_3}{\dim\mathcal H_1}
=
nkm
\]
and
\[
P_2 P_3=P_3P_2=P_1~.
\]
Equivalently, the subspaces
\[
T_{42}\mathcal H_2
\qquad\text{and}\qquad
T_{43}\mathcal H_3
\]
have intersection $T_{41} \mathcal{H}_1$ and the orthogonal complements $T_{42}\mathcal H_2 \ominus T_{41} \mathcal{H}_1$ and $T_{43}\mathcal H_3 \ominus T_{41} \mathcal{H}_1$ are orthogonal to one another. 
\end{proposition}
\begin{proof}
Let $n = \dim(\mathcal{H}_1)$. If the operator maps generate a nondegenerate commuting square, Proposition \ref{NDnormal} gives that the structure of the algebras is characterized up to unitary equivalence  by an identification $\mathcal{H}_4 \cong \mathcal{H}_1 \otimes \mathbb{C}^k \otimes \mathbb{C}^m$, where $k = \dim(\mathcal{H}_2)/n$ and $m = \dim(\mathcal{H}_3)/n$. This immediately imposes the necessary dimension constraint $\dim(\mathcal{H}_4) = nkm = \frac{\dim(\mathcal{H}_2)\dim(\mathcal{H}_3)}{\dim(\mathcal{H}_1)}$. The Proposition also gives that
\begin{align*}
\mathcal{M}_2^{(4)} &= B(\mathcal{H}_1) \otimes B(\mathbb{C}^k) \otimes I_m~, \\
\mathcal{M}_3^{(4)} &= B(\mathcal{H}_1) \otimes U(I_k \otimes B(\mathbb{C}^m))U^\dagger~,
\end{align*}
for some bi-unitary operator $U$ acting on $\mathbb{C}^k \otimes \mathbb{C}^m$.

The projections $P_2,P_3$ lie in the commutant of $\mathcal{M}_2^{(4)}, \mathcal{M}_3^{(4)}$ respectively so can be expressed as
\[
P_2 = I_n \otimes I_k \otimes |e\rangle\langle e| \quad \text{and} \quad P_3 = I_n \otimes U(|f\rangle\langle f| \otimes I_m)U^\dagger
\]
for some unit vectors $e \in \mathbb{C}^m$ and $f \in \mathbb{C}^k$.

The squared cosines of the principal angles between $T_{42}\mathcal{H}_2$ and $T_{43}\mathcal{H}_3$ correspond to the eigenvalues of the product $P_2 P_3 P_2$. Evaluating this yields 
\[
P_2 P_3 P_2 = I_n \otimes (X X^\dagger) \otimes |e\rangle\langle e| \; ,
\]
where $X$ is a $k \times m$ matrix with entries $X_{i \alpha} = (\langle i| \otimes \langle e|) U (|f\rangle \otimes |\alpha\rangle)$. In bases where $e$ and $f$ are chosen as the first basis vectors (which we can denote by the index $1$), $X_{i \alpha} = U_{i 1, 1 \alpha}$. 

The sum of the eigenvalues of $X X^\dagger$ is its trace, $\mathrm{Tr}(X X^\dagger) = \sum_{i, \alpha} |U_{i 1, 1 \alpha}|^2$. By the definition of $U$ being bi-unitary, the reshuffled matrix $\tilde{U}$ defined by $\tilde{U}_{i \alpha, j \beta} = U_{i \beta, j \alpha}$ is also unitary. Summing the absolute squares of the entries in the column fixed by $j=1$ and $\beta=1$ yields $\sum_{i, \alpha} |\tilde{U}_{i \alpha, 1 1}|^2 = 1$. Since $\tilde{U}_{i \alpha, 1 1} = U_{i 1, 1 \alpha}$, this gives $\mathrm{Tr}(X X^\dagger) = 1$.

However, the extremality compatibility condition $T_{41} = T_{42}T_{21} = T_{43}T_{31}$ asserts that the intersection space $T_{41}\mathcal{H}_1$ is contained in both $T_{42}\mathcal{H}_2$ and $T_{43}\mathcal{H}_3$. This implies $P_1 \le P_2$ and $P_1 \le P_3$, which requires $P_2 P_3 P_2 \ge P_1$. Because $P_1 = T_{41}T_{41}^\dagger$ has rank $n$, the operator $P_2 P_3 P_2$ must possess at least $n$ eigenvalues equal to $1$. The non-zero eigenvalues of $P_2 P_3 P_2$ are precisely those of $X X^\dagger$ (each appearing with a degeneracy of $n$), forcing $X X^\dagger$ to have an eigenvalue of $1$. The unit trace condition then demands this to be its unique non-zero eigenvalue. Consequently, the non-zero eigenvalues of $P_2 P_3 P_2$ are strictly $1$ (with multiplicity $n$). This restricts the principal angles to be exactly $0$ and $\pi/2$, establishing that $P_2 P_3 = P_1$ and $(T_{42}\mathcal{H}_2 \ominus T_{41}\mathcal{H}_1) \perp (T_{43}\mathcal{H}_3 \ominus T_{41}\mathcal{H}_1)$.

Conversely, if $\dim(\mathcal{H}_4) = nkm$ and the principal angles are all $0$ and $\pi/2$, then $P_2 P_3 = P_1$. This means the orthogonal complements are strictly perpendicular. We can then naturally define a tensor decomposition $\mathcal{H}_4 \cong \mathcal{H}_1 \otimes \mathbb{C}^k \otimes \mathbb{C}^m$ identifying $T_{41}\mathcal{H}_1 \cong \mathcal{H}_1 \otimes |1\rangle \otimes |1\rangle$, $T_{42}\mathcal{H}_2 \cong \mathcal{H}_1 \otimes \mathbb{C}^k \otimes |1\rangle$, and $T_{43}\mathcal{H}_3 \cong \mathcal{H}_1 \otimes |1\rangle \otimes \mathbb{C}^m$. Picking auxiliary isometries that map into the remaining orthogonal basis vectors immediately yields the factor subalgebras $\mathcal{M}_2^{(4)} = B(\mathcal{H}_1) \otimes B(\mathbb{C}^k) \otimes I_m$ and $\mathcal{M}_3^{(4)} = B(\mathcal{H}_1) \otimes I_k \otimes B(\mathbb{C}^m)$. These inclusion maps satisfy all required operator map properties and inherently form a nondegenerate commuting square corresponding to the trivial identity bi-unitary $U=I$.
\end{proof}
It will be useful to note the following special case:

\begin{corollary}
    Suppose $\mathcal{H}_1 = \mathcal{H}_L \otimes \mathcal{H}_R$, $\mathcal{H}_2 = \hat{\mathcal{H}}_L \otimes \mathcal{H}_R$, $\mathcal{H}_3 = \mathcal{H}_L \otimes \hat{\mathcal{H}}_R$, and $\mathcal{H}_4 = \hat{\mathcal{H}}_L \otimes \hat{\mathcal{H}}_R$ with $\dim \hat{\mathcal{H}}_L = k \dim \mathcal{H}_L$ and $\dim \hat{\mathcal{H}}_R = m \dim \mathcal{H}_R$ for integers $k,m$. Assume we have isometries $T_L : \mathcal{H}_L \to \hat{\mathcal{H}}_L$ and $T_R : \mathcal{H}_R \to \hat{\mathcal{H}}_R$. Then we can use these to define isometries $T_{ij}$ and operator maps $\varphi_{ij}$ satisfying the properties of Proposition 1, and the resulting algebras $\mathcal{M}^{(4)}_i$ form a commuting and co-commuting square. 
\end{corollary}
\begin{proof}
We have
\[
\dim \mathcal{H}_4 = \dim \mathcal{H}_{\hat{L}} \dim \mathcal{H}_{\hat{R}} = {(\dim \mathcal{H}_{\hat{L}} \dim \mathcal{H}_R) (\dim \mathcal{H}_L \dim \mathcal{H}_{\hat{R}}) \over \dim \mathcal{H}_L \dim \mathcal{H}_R} = {\dim \mathcal{H}_2 \dim \mathcal{H}_3 \over \dim \mathcal{H}_1} \; .
\]
Also, 
\[
T_{21} = T_L \otimes I_R\;, \qquad T_{31} = I_L \otimes T_R\;, \qquad T_{42} = I_{\hat{L}} \otimes T_R\;, \qquad T_{43} = T_L \otimes I_{\hat{R}}\;,
\]
and
\[
P_2 =  I_{\hat{L}} \otimes T_R T_R^\dagger \;, \qquad P_3 =   T_L T_L^\dagger \otimes I_{\hat{R}}\;,  \qquad P_1 =   T_L T_L^\dagger \otimes T_R T_R^\dagger \; .
\]
It follows immediately that $T_{42} T_{21} = T_{43} T_{31}$ and that $P_2 P_3 = P_3 P_2 = P_1$. Thus, the necessary and sufficient conditions of Proposition \ref{prop:diamond_commuting_square} are satisfied.
\end{proof}

\section*{Acknowledgments}

This work is supported in part by the National Science and Engineering Research Council of Canada (NSERC) and the Simons Foundation via a Simons Investigator Award.

\appendix

\section{Nondegenerate commuting squares in finite dimensions}
\label{sec:NDcomm}
In this appendix, we describe necessary and sufficient conditions on subalgebras $A \subset B,C \subset D$ of a finite dimensional von Neumann algebra $B(\mathcal{H})$ such that the collection forms a nondegenerate commuting square (equivalently, a commuting and co-commuting square) of von Neumann factors. Finite-dimensional commuting squares have been discussed in \cite{Popa1983,JonesSunder1997,Sano1996,Nicoara2006}.

Since we have a chain of factors $A \subset B \subset D$, we can represent these up to unitary equivalence as
\begin{equation}
A \equiv M_{n}(\mathbb{C}) \otimes I_k \otimes I_m~, \qquad B \equiv M_{n}(\mathbb{C}) \otimes M_{k}(\mathbb{C}) \otimes I_m~, \qquad D \equiv M_{n}(\mathbb{C}) \otimes M_{k}(\mathbb{C}) \otimes M_{m}(\mathbb{C})~. 
\label{cond1}
\end{equation}
Since $A$ is also a factor subalgebra of the remaining algebra $C$, we can represent $C$ in the same basis as
\[ 
C \equiv M_{n}(\mathbb{C}) \otimes \hat{C} \; .
\]
where $\hat{C}$ is a factor subalgebra of $M_{k}(\mathbb{C}) \otimes M_{m}(\mathbb{C})$.

For commuting and co-commuting squares, the dimensions of the algebras satisfy (see \cite{JonesSunder1997,Sano1996} or the appendix of \cite{SahuEtAl2026})
\[
\dim(A) \dim(D) = \dim(B) \dim(C) \; .
\]
so $\hat{C}$ is isomorphic $M_{m}(\mathbb{C})$. We can thus write
\begin{equation}
C \equiv M_{n}(\mathbb{C}) \otimes U (I_k \otimes M_{m}(\mathbb{C})) U^\dagger
\label{cond2}
\end{equation}
for some unitary $U$. We can translate the commuting and co-commuting square conditions to an algebraic condition on $U$, specifically the condition that $U$ is {\bf bi-unitary}:

\begin{definition}[Bi-unitary Matrix]
Let $U \in M_{km}(\mathbb{C})$ be a bipartite unitary matrix with entries $U_{i\alpha, j\beta}$ in the standard tensor basis (where $1 \le i, j \le k$ and $1 \le \alpha, \beta \le m$). 
We define a matrix $U^R$ obtained by a partial transpose, swapping the second and fourth indices of $U$:
\[ (U^R)_{i\beta, j\alpha} := U_{i\alpha, j\beta}~. \]
The matrix $U$ is called \textbf{bi-unitary} if $U^R$ is also a unitary matrix.
\end{definition}

We will show that the square is commuting and co-commuting if and only if the algebras can be represented according to (\ref{cond1}) and (\ref{cond2}) for a bi-unitary matrix $U$.

\begin{proof}
For this discussion, we will assume $A = \mathbb{C} I_D$; in the more general case, $A$ is simply an overall factor that can be added in at each step. Then the trace-preserving conditional expectation onto $A$ is simply $E_A(x) = \tau_D(x) I_D$. The commuting square condition $E_B E_C = E_A$ thus simplifies to the requirement that for all $c \in C$,
\begin{equation}
    E_B(c) = \tau_D(c) I_D~.
    \label{eq1}
\end{equation}
The unique normalized trace on $D = M_k(\mathbb{C}) \otimes M_m(\mathbb{C})$ decomposes as $\tau_D = \tau_k \otimes \tau_m$, where $\tau_k = \frac{1}{k}\Tr$ and $\tau_m = \frac{1}{m}\Tr$.
The conditional expectation $E_B : D \to B$ acts by taking the normalized partial trace over the second subsystem:
\[ E_B(x) = (\id_k \otimes \tau_m)(x) \otimes I_m~. \]
By definition, any element $c \in C$ takes the form $c = U(I_k \otimes y)U^\dagger$ for some arbitrary matrix $y \in M_m(\mathbb{C})$. We can compute the matrix elements of $c$ in the standard tensor product basis $|i, \alpha\rangle$:
\begin{align*}
    c_{i\alpha, l\gamma} &= \sum_{j,\beta, p,\delta} U_{i\alpha, j\beta} \, (I_k \otimes y)_{j\beta, p\delta} \, (U^\dagger)_{p\delta, l\gamma} \\
    &= \sum_{j,\beta, p,\delta} U_{i\alpha, j\beta} \, (\delta_{jp} y_{\beta\delta}) \, U^*_{l\gamma, p\delta}
    \\
    &= \sum_{j=1}^k \sum_{\beta, \delta=1}^m U_{i\alpha, j\beta} \, y_{\beta\delta} \, U_{l\gamma, j\delta}^*~.
\end{align*}
Applying the partial trace map $(\id_k \otimes \tau_m)$ traces out the Greek indices of the ambient space (setting $\alpha = \gamma$ and summing, then dividing by $m$). This yields a $k \times k$ matrix whose $(i, l)$-th entry is:
\begin{equation}
    [E_B(c)]_{i, l} = \frac{1}{m} \sum_{\alpha=1}^m c_{i\alpha, l\alpha} = \frac{1}{m} \sum_{\alpha, \beta, \delta=1}^m \sum_{j=1}^k U_{i\alpha, j\beta} \, U^*_{l\alpha, j\delta} \, y_{\beta\delta}~.
    \label{eq2}
\end{equation}
Conversely, the right-hand side of Equation \ref{eq1} evaluates to:
\begin{equation}
    [\tau_D(c) I_k]_{i, l} = \tau_m(y) \delta_{i,l} = \left( \frac{1}{m} \sum_{\beta=1}^m y_{\beta\beta} \right) \delta_{i,l} = \frac{1}{m} \sum_{\beta, \delta=1}^m y_{\beta\delta} \, \delta_{\beta,\delta} \, \delta_{i,l}~.
    \label{eq3}
\end{equation}
For $B$ and $C$ to form a commuting square, Equation \ref{eq2} must equal Equation \ref{eq3} for every possible matrix $y \in M_m(\mathbb{C})$. This occurs if and only if the coefficients of $y_{\beta\delta}$ match exactly for all terms:
\begin{equation}
    \sum_{j=1}^k \sum_{\alpha=1}^m U_{i\alpha, j\beta} U^*_{l\alpha, j\delta} = \delta_{i,l} \delta_{\beta,\delta}~.
\end{equation}
This is equivalent to
\[ U^R (U^R)^\dagger = I_{km} \]
which is the condition that $U$ is bi-unitary.

We now check the co-commuting condition. In the reduced case
\(A=\mathbb C I_D\), the relative commutants inside \(D=M_k\otimes M_m\) are
\[
B'\cap D=I_k\otimes M_m~,
\qquad
C'\cap D=U(M_k\otimes I_m)U^*~.
\]
Thus the co-commuting condition is
\[
E_{I_k\otimes M_m}
E_{U(M_k\otimes I_m)U^*}
=
E_{\mathbb C I}~.
\]
Let
\[
z=U(x\otimes I_m)U^*
\]
with \(x\in M_k(\mathbb C)\). Its matrix entries are
\[
z_{i\alpha,l\gamma}
=
\sum_{j,p=1}^k\sum_{\beta=1}^m
U_{i\alpha,j\beta}\,x_{jp}\,U^*_{l\gamma,p\beta}~.
\]
Applying the conditional expectation onto \(I_k\otimes M_m\) gives
\[
[E_{I_k\otimes M_m}(z)]_{\alpha,\gamma}
=
\frac1k
\sum_{i=1}^k
z_{i\alpha,i\gamma}
=
\frac1k
\sum_{i,j,p,\beta}
U_{i\alpha,j\beta}\,
U^*_{i\gamma,p\beta}\,
x_{jp}~.
\]
On the other hand,
\[
\tau_D(z)I_m=\tau_k(x)I_m~,
\]
whose \((\alpha,\gamma)\)-entry is
\[
\frac1k\sum_j x_{jj}\delta_{\alpha,\gamma}~.
\]
Since this must hold for all \(x\in M_k(\mathbb C)\), the co-commuting
condition is equivalent to
\[
\sum_{i=1}^k\sum_{\beta=1}^m
U_{i\alpha,j\beta}\,
U^*_{i\gamma,p\beta}
=
\delta_{j,p}\delta_{\alpha,\gamma}~.
\]
In terms of \(U^R\), this is precisely
\[
(U^R)^*U^R=I_{km}~.
\]
The original commuting-square condition gave
\[
U^R(U^R)^*=I_{km}~.
\]
Since \(U^R\) is a square matrix, either one of these equations implies the
other. Hence the square is commuting and co-commuting if and only if \(U^R\)
is unitary, i.e. if and only if \(U\) is bi-unitary.

\end{proof}

\bibliographystyle{jhep}
\bibliography{factor_code_networks_references}
\end{document}